\documentclass[journal,onecolumn]{IEEEtran}

\usepackage{amsmath,amsthm,amssymb,amsfonts}
\usepackage{array,tabularx,booktabs,multirow}
\usepackage{xcolor,tikz}
\usepackage{hyperref}
\hypersetup{hidelinks,colorlinks=true,linkcolor=blue,citecolor=blue,urlcolor=blue}
\usepackage[noadjust]{cite}

\makeatletter
\newtheoremstyle{mythm}{3pt}{3pt}{}{16pt}{\bfseries}{:}{.5em}{}
\theoremstyle{mythm}
\newtheorem{theorem}{Theorem}
\newtheorem{definition}{Definition}
\newtheorem{lemma}{Lemma}
\newtheorem{corollary}{Corollary}
\newtheorem{proposition}{Proposition}

\newtheorem{example}{Example}
\newtheorem{remark}{Remark}

\newtheorem{construction}{Construction}

\makeatother

\DeclareMathOperator{\supp}{Supp}
\DeclareMathOperator{\elp}{ELP}

\begin{document}
\title{Entropy/Length Profile Bounds for the Conversion of (Non)linear Codes in the Merge Regime}
\author{
    Songping~Ge,
    Han~Cai,~\IEEEmembership{Member,~IEEE},
    and~Xiaohu~Tang,~\IEEEmembership{Fellow,~IEEE}
    \thanks{Songping Ge and Xiaohu Tang are with the School of Mathematics, Southwest Jiaotong University, Chengdu, 610031, China. Xiaohu Tang is also with  the Information Coding and Transmission Key Lab of Sichuan Province, CSNMT Int. Coop. Res. Centre (MoST), Southwest Jiaotong University, Chengdu, 610031, China  (e-mail: gesongping@my.swjtu.edu.cn; xhutang@swjtu.edu.cn).}
    \thanks{Han Cai is with the Information Coding and Transmission Key Lab of Sichuan Province, CSNMT Int. Coop. Res. Centre (MoST), Southwest Jiaotong University, Chengdu, 610031, China (e-mail: hancai@aliyun.com).}
}
\maketitle

\begin{abstract}
Code conversion changes the parameters of encoded data while avoiding complete decoding and re-encoding. Existing general bounds largely rely on the algebraic structure of linear codes. This paper develops an information-theoretic framework for nonlinear code conversion in the merge regime. Using the entropy/length profile (ELP), we derive upper bounds on the number of unchanged symbols and lower bounds on the number of read symbols for arbitrary $(n,M)_q$ codes, without assuming linearity or requiring $\log_q M$ to be an integer. We also obtain joint bounds for multiple initial codes and consequences involving minimum distance. For linear codes, the proposed bounds reduce to the corresponding generalized Hamming weight (GHW) bounds. The construction based on Hamming codes converts Hamming codes of redundancy two into a Hamming code of any larger redundancy and simultaneously attains the general bounds on unchanged and read symbols. The construction based on quasigroups shows that new bounds remain tight for genuinely nonlinear codes of nonintegral dimension and at a nonintegral entropy threshold, which, to the best of our knowledge, is the first construction of convertible codes for nonlinear codes.
\end{abstract}


\section{Introduction}\label{sec1}

Distributed storage systems at large scale use erasure codes to protect data against device failures with substantially less storage overhead than replication \cite{ghemawat2003google,huang2012erasure}. A file is divided into $k$ information symbols, encoded into a length $n$ codeword, and stored across different devices. The code rate $k/n$ determines the balance between redundancy and storage efficiency. In practice, however, failure rates, hardware configurations, and patterns of data access change over time. A code rate chosen for one operating condition may provide insufficient protection when failures become more frequent or waste storage when fewer redundant symbols are needed. Adapting redundancy to device reliability can therefore improve storage efficiency. The system study in \cite{kadekodi2019cluster} reports substantial storage savings from such adaptation.

Changing the code in a running system by first decoding the stored data and then encoding it again may require many symbols to be accessed and rewritten. Convertible codes were introduced to perform this change more efficiently while preserving the encoded information \cite{maturana2020convertible,maturana2022convertible}. Two basic settings are commonly considered. In the merge regime, several initial codewords are combined into one final codeword. In the split regime, one initial codeword is divided into several final codewords. During conversion, some initial symbols may remain unchanged, while the other final symbols are computed from symbols read from the initial codewords. The read cost counts the initial symbols that are accessed, the write cost counts the newly written final symbols, and their sum is the access cost. A related but different measure is the bandwidth cost, which counts the amount of information transferred and may allow parts of a symbol to be downloaded \cite{maturana2022bandwidth,maturana2023bandwidth}.

Early work established the basic limits of code conversion and constructed linear MDS convertible codes with optimal access cost in the merge regime \cite{maturana2020accessoptimal,maturana2022convertible}. Subsequent studies reduced the required field size, extended the available parameter ranges, and developed further MDS constructions \cite{chopra2024low,ramkumar2026mds}. The irregular setting, in which the initial codes may have different lengths, dimensions, and redundancies, was studied in \cite{ge2026mds}. That work derived bounds on access cost, characterized optimal linear MDS conversions through parity-check matrices, and gave explicit constructions. Bandwidth has been considered separately for both merge and split conversion, leading to additional lower bounds and code constructions \cite{maturana2023bandwidth,wang2026lower,singhvi2026lower}.

Code conversion has also been studied for codes with locality. Initial results connected local repair with efficient conversion \cite{maturana2023locally}. Later work obtained locally repairable convertible codes with optimal access cost, improved lower bounds, and constructions for generalized merge conversion into locally repairable codes \cite{kong2024locally,ge2026locally,shi2026bounds}. These results show that conversion cost is closely related to the structural properties of the code family under consideration. However, most of the available bounds and constructions concern MDS codes, locally repairable codes, or closely related linear settings.

More recently, Gruica, Jany, and Kruglik developed bounds for arbitrary scalar linear codes in the merge regime \cite{gruica2026linear}. Their approach uses GHWs to control the supports of subcodes forced by unchanged and read coordinates. It recovers earlier bounds in important special cases and applies beyond MDS and locally repairable codes. Nevertheless, the argument remains linear. It uses dimensions of vector spaces, linear subcodes, and GHWs with integer indices. It does not directly cover a nonlinear initial or final code, a nonlinear conversion map, or a code of size $M$ for which $\log_q M$ is not an integer.

However, all known results on convertible codes are focus on linear codes, and much less is known about the more general nonlinear case. A nonlinear code need not be closed under addition, and a set obtained by fixing part of a conversion need not have the algebraic structure used in linear proofs. Moreover, GHWs of nonlinear codes arise at real entropy levels rather than only at integer dimensions. The conditional entropy/length profile (ELP) introduced in \cite{reuven1998entropy} is well suited to this setting. Its jump points give GHWs for nonlinear codes \cite{reuven1999generalized}, while the full profile retains the intermediate entropy levels that an integer index would omit. In a merge conversion, fixing selected initial codewords and conditioning on unchanged coordinates reveal how much information must still be represented by the written coordinates. We use this observation to derive bounds directly from entropy, without assuming linearity or an integer dimension. Thus, known methods and analyses generally fail for nonlinear codes, and it may be interesting to consider new methods for convertible codes.

\emph{Scope and contributions:} We consider information-preserving merge conversions involving arbitrary scalar codes. The initial codes may have different parameters, and neither the codes nor the conversion map are assumed to be linear. Our main contributions are as follows.
\begin{itemize}
	\item We formulate an inverse $\delta_\rho^{\elp}$ of the ELP indexed by real entropy levels and establish the properties needed for conversion. This formulation retains entropy levels that are absent from the hierarchy of GHWs with integer indices.
	\item We derive individual and joint upper bounds on unchanged symbols. The proofs use conditional entropy and remain valid although the ELP need not be monotone under code inclusion.
	\item We derive a lower bound on read symbols for arbitrary real $\log_q M$. A simplified form is
	\begin{equation*}
		|\mathcal R_i|\ge
		\left\lceil\log_qM_{I_i}-\theta_i\right\rceil,
	\end{equation*}
	where $\theta_i$ is the maximum residual entropy over sets of $|\mathcal U_i|$ final coordinates, capped at $k_i$.
	\item We recover the general linear bounds of \cite{gruica2026linear} as special cases and identify the additional conditioning required in the nonlinear setting. Table \ref{tab1} summarizes the theoretical results. Detailed comparisons with \cite{gruica2026linear} appear in Section \ref{sec5}.
	\item We introduce a construction of convertible codes via Hamming codes that simultaneously attains the general bounds on the numbers of unchanged and read symbols. Moreover, we also propose a construction via quasigroups, which is the first construction of convertible codes for nonlinear codes and simultaneously attains the same general bounds on the numbers of unchanged and read symbols. 
\end{itemize}

\begin{table*}[t]
	\centering
	\caption{Comparison of the present bounds with the corresponding linear bounds in \cite{gruica2026linear}.}
	\label{tab1}
	\renewcommand{\arraystretch}{1.3}
	\setlength{\tabcolsep}{1pt}
	\renewcommand{\tabularxcolumn}[1]{m{#1}}
	\begin{tabularx}{\textwidth}{
		>{\centering\arraybackslash}m{0.06\textwidth}
		>{\centering\arraybackslash}m{0.32\textwidth}
		>{\centering\arraybackslash}m{0.10\textwidth}
		>{\centering\arraybackslash}m{0.29\textwidth}
		>{\centering\arraybackslash}X
		>{\centering\arraybackslash}m{0.05\textwidth}}
		\toprule
		Result & Bound & Difference & Linear Bound \cite{gruica2026linear} & Difference  & \cite{gruica2026linear} \\
		\midrule
		\midrule
		Lem. \ref{lem6}
		&
		$|\mathcal U_i|\le\min_{\mathbf v} \{n_{I_i},n_F-\delta_{k_{-i}}^{\elp} (\mathcal C_{i,\mathbf v})\}$
		&
		$\mathcal C_{i,\mathbf v}$ not cosets
		&
		$|\mathcal U_i|\le\min\{n_{I_i},n_F-d_{k_{-i}}(\mathcal C_{i,\mathbf 0})\}$
		&
		$\mathcal C_{i,\mathbf 0}$ cosets
		&
		Prop.~3\\
		\midrule
		Lem. \ref{lem7}
		&
		$|\mathcal U_i|\le\min\{n_{I_i},n_F-\delta_{k_{-i}}^{\elp}(\mathcal C_F)\}$
		&
		Non-monotonicity
		&
		$|\mathcal U_i|\le\min\{n_{I_i},n_F-d_{k_{-i}}(\mathcal C_F)\}$
		&
		$d_{k_{-i}}(\mathcal C_F) \le d_{k_{-i}}(\mathcal C_{i,\mathbf 0})$
		&
		Prop.~3\\
		\midrule
		Lem. \ref{lem8}
		&
		$\sum_{j\in I}|\mathcal U_j|\le\min\{\sum_{j\in I}n_{I_j}, n_F-\delta_{k_{-I}}^{\elp}(\mathcal C_F)\}$
		&
		--
		&
		-- & -- & --\\
		\midrule
		Cor. \ref{cor1}
		&
		$|\mathcal U_i|\le\min\{n_{I_i}, n_F-d(\mathcal C_F)-\lceil k_{-i}\rceil+1\}$
		&
		$k_{-i}\notin\mathbb Z$
		&
		$|\mathcal U_i|\le\min\{n_{I_i}, n_F-d(\mathcal C_F)-k_{-i}+1\}$
		&
		$k_{-i}\in\mathbb Z$
		&
		Cor.~4\\
		\midrule
		Lem. \ref{lem9}
		&
		$|\mathcal U|\le \left\lfloor \frac{tn_F-t(d(\mathcal C_F)-1)-\sum_i\lceil k_i\rceil}{t-1} \right\rfloor$
		&
		Ceiling
		&
		$|\mathcal U|\le \frac{tn_F-k_F-t(d(\mathcal C_F)-1)}{t-1}$
		&
		$\sum_i k_i=k_F$
		&
		Prop.~7\\
		\midrule
		Lem. \ref{lem10}
		&
		$|\mathcal R_i|\ge \max\{0,\lceil k_i-|\mathcal U_i|\rceil, \max_{0<r\le k_i}g_i(r)\}$
		&
		$r\in\mathbb R$
		&
		$\displaystyle |\mathcal R_i|\ge \max_{1\le j\le k_i} \{k_i-j+1-(|\mathcal U_i|-d_j+1)^+\}$
		&
		$j\in\mathbb Z$
		&
		Prop.~9\\
		\midrule
		Thm. \ref{thm1}
		&
		$\displaystyle |\mathcal R_i|\ge \lceil k_i-\theta_i\rceil$
		&
		$\theta_i\in\mathbb R$
		&
		$\displaystyle |\mathcal R_i|\ge k_i-j_i^*$
		&
		$j_i^*\in\mathbb Z$
		&
		Thm.~10\\
		\bottomrule
	\end{tabularx}
\end{table*}

\emph{Organization:} Section \ref{sec2} introduces the ELP and the model for merge conversion. Sections \ref{sec3} and \ref{sec4} derive bounds on unchanged and read symbols, respectively. Section \ref{sec5} specializes the results to linear codes and compares them with \cite{gruica2026linear}. Section \ref{sec6} gives conversions based on Hamming codes and quasigroups. Section \ref{sec10} concludes the paper.

\section{Preliminary}\label{sec2}

For a positive integer $a$, let $[a]=\{1,2,\ldots,a\}$. For $J\subseteq[n]$, write $J^c=[n]\setminus J$.

Let $\mathcal C\subseteq\mathcal A^n$ be a code of length $n$ and size $M$, where $|\mathcal A|=q$. We denote its parameters by $(n,M)_q$ and do not assume that $\mathcal C$ is linear. If the minimum Hamming distance of $\mathcal C$ is $d$, we denote its parameters by $(n,M,d)_q$. If $\mathcal C$ is linear, we denote its parameters by $[n,\log_q M,d]_q$. Let $X^{(\mathcal C)}$ be a random codeword chosen uniformly from $\mathcal C$. For $\mathbf c\in\mathcal A^n$ and $J\subseteq[n]$, let $\mathbf c|_J$ denote the restriction of $\mathbf c$ to the coordinates in $J$. The coordinate projection of $\mathcal C$ onto $J$ is
\begin{equation*}
    \mathcal C|_J
    =\{\mathbf c|_J:\mathbf c\in\mathcal C\}.
\end{equation*}
Correspondingly, $X_J^{(\mathcal C)}\triangleq X^{(\mathcal C)}|_J$ is the random variable obtained by restricting the uniformly chosen codeword $X^{(\mathcal C)}$ to the coordinates in $J$. Thus, $\mathcal C|_J$ is the set of all projected codewords, whereas $X_J^{(\mathcal C)}$ is a random variable taking values in $\mathcal C|_J$ with the distribution induced by the uniform distribution on $\mathcal C$. When the underlying code is clear from the context, we write $X$ and $X_J$ in place of $X^{(\mathcal C)}$ and $X_J^{(\mathcal C)}$, respectively.

\begin{definition}[\cite{reuven1999generalized}]\label{def1}
	The support of an $(n,M)_q$ code $\mathcal C$ is
	\begin{equation}\label{def_supp}
		\supp(\mathcal C)=
		\{i\in[n]: \text{ exist }\mathbf c=(c_1,\dots,c_n),\mathbf b=(b_1,\dots,b_n)\in\mathcal C \text{ with } c_i\ne b_i \}.
	\end{equation}
\end{definition}

\subsection{Entropy/Length Profiles}

We recall the conditional entropy/length profile (ELP), which underlies the GHW hierarchy for nonlinear codes. All entropies are measured in base $q$.

\begin{definition}[{\cite[Definition 4]{reuven1998entropy}}]\label{def2}
	The conditional entropy/length profile of an $(n,M)_q$ code $\mathcal C$ is the sequence $\{h_i(\mathcal C)\}_{i=0}^{n}$, where
	\begin{equation*}
		h_i(\mathcal C) = \underset{J\subseteq[n], |J|=i}{\max} H(X_J|X_{J^c})
	\end{equation*}
	with the convention $h_0(\mathcal C)=0$. In addition, $h_n(\mathcal C) = \log_q M$.
\end{definition}

The ELP satisfies the following fundamental property, which will be used repeatedly in our subsequent derivations.

\begin{lemma}\label{lem1}
	If $h_i(\mathcal C)\ge\rho$, then $i\ge\rho$ for any $\rho\in\mathbb R$ with $0\le\rho\le\log_q M$.
\end{lemma}
\begin{IEEEproof}
	Suppose, for contradiction, that $i<\rho$. Then,
	\begin{align*}
		h_i(\mathcal C) 
		&= \underset{J\subseteq[n], |J|=i}{\max} H(X_J|X_{J^c}) \\
		&\le \underset{J\subseteq[n], |J|=i}{\max} H(X_J) \\
		&\le \min\{|J|,H(X)\} \\
		&= \min\{i,\log_q M\} \\
		&< \rho,
	\end{align*}
	which contradicts $h_i(\mathcal C)\ge\rho$.	
\end{IEEEproof}

Another important feature of the ELP is its monotonicity with bounded increments, as stated in the following lemma from \cite{reuven1998entropy}.

\begin{lemma}[{\cite[Lemma 3]{reuven1998entropy}}]\label{lem2}
	If $\mathcal C$ is an $(n,M)_q$ code, then for any $i, 0\le i\le n-1$,
	\begin{equation*}
		0\le h_{i+1}(\mathcal C) - h_i(\mathcal C) \le 1.
	\end{equation*}
\end{lemma}

\subsection{ELP Inverse at Real Entropy Levels}

We next recall the definition of GHWs for nonlinear codes proposed in \cite{reuven1999generalized}. This definition is based on the jump points of the ELP sequence. We then give an equivalent formulation indexed by real entropy levels that is more convenient for conversion bounds.

\begin{definition}[{\cite[Definition 2]{reuven1999generalized}}]\label{def3}
	Let $\{h_i\}_{i=0}^{n}$ be the ELP of $\mathcal C$, and write its jump set as
	\begin{equation*}
		\mathcal J_{\elp}=\{i\in[n]:h_i>h_{i-1}\}=\{j_1,\ldots,j_m\},
	\end{equation*}
	where $j_1<\cdots<j_m$. The generalized Hamming weight (GHW) of order $r$ is $d_r^{\elp}(\mathcal C)=j_r$.
\end{definition}

We use the following inverse of the ELP in the conversion bounds.

\begin{definition}[ELP Inverse]\label{def4}
	Let $\mathcal C$ be an $(n,M)_q$ code. For any $\rho\in\mathbb R$ satisfying $0\le \rho\le \log_q M$, define
	\begin{equation*}
		\delta^{\elp}_\rho(\mathcal C) = \min \{ i\in[n] : h_i(\mathcal C)\ge\rho \}
	\end{equation*} 
	with the convention $\delta^{\elp}_0(\mathcal C)=0$.
\end{definition}

The next result shows that the positive values of $\delta^{\elp}_\rho$ are exactly the ELP jump positions.

\begin{theorem}
	Let $\mathcal C$ be an $(n,M)_q$ code. Then,
	\begin{equation*}
		\mathcal J_{\elp}=\{\delta^{\elp}_\rho(\mathcal C):0<\rho\le\log_qM\}.
	\end{equation*}
\end{theorem}

\begin{IEEEproof}
	Fix $0<\rho\le\log_q M$ and let
	\begin{equation*}
		i=\delta^{\elp}_\rho(\mathcal C)= \min \{ j\in[n] : h_j(\mathcal C)\ge\rho \}.
	\end{equation*}
	Since $\rho>0=h_0(\mathcal C)$, we have $i\ge1$. If $i\notin\mathcal J_{\elp}$, then $h_i(\mathcal C)=h_{i-1}(\mathcal C)\ge\rho$ by Definition \ref{def3}, contradicting the minimality of $i$. Hence
	\begin{equation*}
		\{\delta^{\elp}_\rho(\mathcal C):0<\rho\le\log_qM\}\subseteq\mathcal J_{\elp}.
	\end{equation*}
	
	Conversely, let $j_r\in\mathcal J_{\elp}$ and take $\rho=h_{j_r}(\mathcal C)$. Lemma \ref{lem2} and $h_{j_r-1}(\mathcal C)<h_{j_r}(\mathcal C)$ give $h_i(\mathcal C)<\rho$ for every $i<j_r$. Therefore, $j_r=\delta^{\elp}_\rho(\mathcal C)$ which implies
	\begin{equation*}
		\{\delta^{\elp}_\rho(\mathcal C):0<\rho\le\log_qM\}\supseteq\mathcal J_{\elp}.
	\end{equation*}
	
	We complete the proof.
\end{IEEEproof}

The generalized inverse is monotone in its entropy argument as follows.

\begin{lemma}\label{lem3}
	Let $\mathcal C$ be an $(n,M)_q$ code. Then, $\delta^{\elp}_{\rho_2}(\mathcal C)\ge \delta^{\elp}_{\rho_1}(\mathcal C)$ for any $0\le \rho_1\le \rho_2\le \log_q M$.
\end{lemma}

\begin{IEEEproof}
	If $h_i(\mathcal C)\ge\rho_2$, then $h_i(\mathcal C)\ge\rho_1$. Hence,
	\begin{equation*}
		\{ i\in[n] : h_i(\mathcal C)\ge\rho_2 \} \subseteq \{ i\in[n] : h_i(\mathcal C)\ge\rho_1 \}.
	\end{equation*}
	Taking the minimum value proves the result.
\end{IEEEproof}

\begin{remark}
	Although $h_i(\mathcal C)$ increases by at most one when $i$ increases by one in Lemma \ref{lem2}, this property does not give a similar bound on $\delta_\rho^{\elp}(\mathcal C)$. Indeed, a small increase in $\rho$ may move $\delta_\rho^{\elp}(\mathcal C)$ to a later ELP jump and therefore increase it by several coordinates. On the other hand, $\delta_\rho^{\elp}(\mathcal C)$ may remain unchanged while $\rho$ varies between two entropy levels reached at the same ELP jump. Therefore, for real $\rho$, we use only the monotonicity established in Lemma \ref{lem3}.
\end{remark}

The following consequence involving minimum distance will be useful below.

\begin{lemma}\label{lem4}
	Let $\mathcal C$ be an $(n,M)_q$ code with minimum Hamming distance $d(\mathcal C)$. Then, for every $0<\rho\le\log_q M$,
	\begin{equation*}
		\delta_{\rho}^{\elp}(\mathcal C)
		\ge d(\mathcal C)+\lceil\rho\rceil-1.
	\end{equation*}
\end{lemma}

\begin{IEEEproof}
	If $i<d(\mathcal C)$, then any two codewords $\mathbf a, \mathbf b$ such that $\mathbf a|_{J^c}=\mathbf b|_{J^c}$, where $|J|=i$, must satisfy $\mathbf a=\mathbf b$. Otherwise their Hamming distance would be at most $i<d(\mathcal C)$. Hence, for $0\le i\le d(\mathcal C)-1$,
	\begin{equation*}
		h_i(\mathcal C) = \underset{J\subseteq[n], |J|=i}{\max} H(X_J|X_{J^c}) = 0.
	\end{equation*}
	For $i\ge d(\mathcal C)-1$, Lemma \ref{lem2} shows that
	\begin{equation}\label{eqn1}
		h_i(\mathcal C)
		\le h_{d(\mathcal C)-1}(\mathcal C)+ i-d(\mathcal C)+1
		= i-d(\mathcal C)+1.
	\end{equation}
	If $h_i(\mathcal C)\ge\rho>0$, then $h_i(\mathcal C)>0$, so $i\ge d(\mathcal C)$. The inequality \eqref{eqn1} implies $$i\ge d(\mathcal C)+\rho-1.$$ Since $i$ and $d(\mathcal C)$ are integers, this is equivalent to $$i\ge d(\mathcal C)+\lceil\rho\rceil-1.$$ Taking the smallest such $i$ satisfying $h_i(\mathcal C)\ge\rho$ proves the claim.
\end{IEEEproof}

We also show the following discrete version of monotonicity. Under the linear specialization in Proposition \ref{prop1}, it becomes the usual strict growth property of GHWs for linear codes.

\begin{lemma}\label{lem5}
	Let $\mathcal C$ be an $(n,M)_q$ code. For any integers $1\le\rho_1\le \rho_2\le\log_q M$,
	\begin{equation*}
		\delta_{\rho_1}^{\elp}(\mathcal C)\le \delta_{\rho_2}^{\elp}(\mathcal C)-(\rho_2-\rho_1).
	\end{equation*}
\end{lemma}

\begin{IEEEproof}
	Let $i_0=\delta_{\rho_1}^{\elp}(\mathcal C)\le n$ and $d=\rho_2-\rho_1$. By Definition \ref{def4},
	\begin{align*}
		h_{i_0-1}(\mathcal C) &< \rho_1, \\
		h_{i_0}(\mathcal C) &\ge \rho_1.
	\end{align*}
	We claim that $i_0+d\le n$. Otherwise, $i_0+d>n$ if and only if $n-i_0+1\le d$. Applying Lemma \ref{lem2} from index $i_0-1$ to index $n$ shows that
	\begin{align*}
		h_n(\mathcal C)
		&\le h_{i_0-1}(\mathcal C)+n-i_0+1\\
		&<\rho_1+d \\
		&=\rho_2,
	\end{align*}
	contrary to $h_n(\mathcal C)=\log_qM\ge\rho_2$. Since $i_0+d-1\le n-1$, another application of Lemma \ref{lem2} gives
	\begin{equation*}
		h_{i_0+d-1}(\mathcal C)
		\le h_{i_0-1}(\mathcal C)+d
		<\rho_1+d=\rho_2.
	\end{equation*}
	Thus the ELP has not reached $\rho_2$ by index $i_0+d-1$. Definition \ref{def4} therefore yields
	\begin{align*}
		\delta_{\rho_2}^{\elp}(\mathcal C)
		&\ge i_0+d\\
		&=\delta_{\rho_1}^{\elp}(\mathcal C)+(\rho_2-\rho_1),
	\end{align*}
	which is equivalent to the claimed inequality.
\end{IEEEproof}

\subsection{Merge Conversion}

In this paper, we study the convertible codes in the merge regime following the theoretical framework introduced in \cite{maturana2022convertible} and generalized in \cite{ge2026mds}.

\begin{definition}
	A convertible code in the merge regime consists of
	\begin{enumerate}
		\item $t$ initial codes $\{\mathcal C_{I_i}\}_{i=1}^{t}$, where $\mathcal C_{I_i}$ is an $(n_{I_i},M_{I_i})_q$ code for all $i\in[t]$,
		\item a final code $\mathcal C_F$ with parameters $(n_F,M_F)_q$,
		\item a conversion map $\sigma:\prod_{i=1}^{t} \mathcal C_{I_i}\to \mathcal C_F$.
	\end{enumerate}
\end{definition}

Let $\mathcal C_{\mathbf I}=(\mathcal C_{I_1},\ldots,\mathcal C_{I_t})$. We denote the resulting convertible code by $(\mathcal C_{\mathbf I},\mathcal C_F,\sigma)$. A conversion is information preserving if $\sigma$ is bijective, in which case $M_F=\prod_{i=1}^{t}M_{I_i}$. Throughout the remainder of the paper, all conversions are assumed to be information preserving. We assume $M_{I_i}\ge2$ for every $i\in[t]$ to exclude initial codes that carry no information. The initial codes may have different lengths, sizes, distances, and algebraic structures.

We distinguish three types of symbols in a conversion. An unchanged symbol is an initial symbol that appears at the same indexed coordinate in the final codeword. A read symbol is an initial symbol accessed to compute the new final symbols. A written symbol is a final symbol produced during conversion. 

To distinguish coordinates belonging to different codewords, the coordinates of the $i$th initial code are indexed by
\begin{equation*}
	\mathcal I_i=\{i\}\times[n_{I_i}],\ i\in[t].
\end{equation*}
For a merge conversion, let $\mathcal U_i\subseteq\mathcal I_i$ and $\mathcal R_i\subseteq\mathcal I_i$ denote the unchanged and read coordinates, respectively. These sets may overlap because an unchanged symbol may also be read. Put
\begin{equation*}
	\mathcal U=\bigsqcup_{i=1}^t\mathcal U_i,\
	\mathcal R=\bigsqcup_{i=1}^t\mathcal R_i,\
	\mathcal W=\{t+1\}\times[n_F-|\mathcal U|],\
	\Omega_F=\mathcal U\sqcup\mathcal W.
\end{equation*}
Thus the coordinates of the final code are indexed by the disjoint set $\Omega_F$, which has size $n_F$. The sets $\mathcal U_i$ and $\mathcal R_i$ describe the conversion if
\begin{equation}\label{eqn2}
	\sigma(\mathbf c)|_{\mathcal U}=\mathbf c|_{\mathcal U},
	\qquad
	\sigma(\mathbf c)|_{\mathcal W}=\psi(\mathbf c|_{\mathcal R})
\end{equation}
for every input tuple $\mathbf c$ and some deterministic map $\psi$. For $i\in[t]$ and $I\subseteq[t]$, define
\begin{equation*}
	k_i=\log_qM_{I_i},\ k_F=\log_qM_F,\ k_{-I}=k_F-\sum_{i\in I}k_i.
\end{equation*}
For $I=\{i\}$, we write $k_{-i}$ in place of $k_{-I}$, i.e.,
\begin{equation}\label{eqnk-i}
	k_{-i} = k_F - k_i.
\end{equation}

For systematic linear initial codes, reading all information symbols and then re-encoding gives a conversion with total read cost $k_F$. More generally, each $q$-ary read symbol carries at most one unit of $q$-ary entropy. Hence, if the total read cost is smaller than $k_F$, the read symbols alone cannot determine all the input information, and the conversion must make essential use of unchanged symbols. This motivates the following definition.

\begin{definition}\label{def5}
	A merge conversion that preserves information is nontrivial if its total read cost is strictly smaller than the sum of the dimensions of the initial codes, namely,
	\begin{equation*}
		\sum_{i=1}^{t}|\mathcal R_i|<k_F.
	\end{equation*}
\end{definition}

The preceding definitions separate the two resources studied in this paper. Unchanged symbols reduce the write cost, whereas read sets quantify the information accessed to compute the written symbols. The next section bounds the number of unchanged symbols.

\section{Bounds on Unchanged Symbols}\label{sec3}

\subsection{Individual Bounds}

We begin with an upper bound for each initial code. The first form uses the ELP of the subset of the final code obtained by fixing one initial codeword. For $\mathbf v\in\mathcal C_{I_i}$, write
\begin{equation*}
	\mathcal C_{i,\mathbf v}
	=\left\{\sigma(\mathbf c):
	\mathbf c\in \prod_{j=1}^{i-1}\mathcal C_{I_j} \times \{\mathbf{v}\} \times \prod_{j=i+1}^{t}\mathcal C_{I_j} \right\}.
\end{equation*}

\begin{lemma}\label{lem6}
	Let $(\mathcal C_{\mathbf{I}},\mathcal C_F,\sigma)$ be a convertible code. For every $i\in[t]$,
	\begin{equation*}
		|\mathcal U_i| \le \min_{\mathbf{v}\in\mathcal C_{I_i}} \left\{ n_{I_i}, n_F-\delta^{\elp}_{k_{-i}}(\mathcal C_{i,\mathbf v}) \right\},
	\end{equation*}
	where $k_{-i}$ is defined by \eqref{eqnk-i}.
\end{lemma}

\begin{IEEEproof}
	Without loss of generality, assume $i=1$ and fix $\mathbf v\in\mathcal C_{I_1}$. Since $\sigma$ is bijective,
	\begin{equation*}
		\log_q|\mathcal C_{1,\mathbf v}| = k_{-1}.
	\end{equation*}
	Thus,
	\begin{equation*}
		\Pr\!\left(X^{(\mathcal C_{1,\mathbf v})}=\mathbf c\right)
		=\frac{1}{|\mathcal C_{1,\mathbf v}|}
	\end{equation*}
	for every $\mathbf c\in\mathcal C_{1,\mathbf v}$, and 
	\begin{equation*}
		H(X^{(\mathcal C_{1,\mathbf v})})=k_{-1}.
	\end{equation*}
	The coordinates in $\mathcal U_1$ are constant on $\mathcal C_{1,\mathbf v}$, by \eqref{def_supp},
	\begin{equation*}
		|\supp(\mathcal C_{1,\mathbf v})| \le n_F - |\mathcal U_1|.
	\end{equation*}
	All coordinates outside the support are constant. Therefore,
	\begin{align*}
	&h_{|\supp(\mathcal C_{1,\mathbf v})|}
	(\mathcal C_{1,\mathbf v})\\
	&\ge H\left(X^{(\mathcal C_{1,\mathbf v})}_{\supp(\mathcal C_{1,\mathbf v})}
	\middle|X^{(\mathcal C_{1,\mathbf v})}_{\Omega_F\setminus\supp(\mathcal C_{1,\mathbf v})}\right) \\
	&= H\left( X^{(\mathcal C_{1,\mathbf v})} \right) - H\left( X^{(\mathcal C_{1,\mathbf v})}_{\Omega_F\setminus\supp(\mathcal C_{1,\mathbf v})} \right) \\
	&= H\left( X^{(\mathcal C_{1,\mathbf v})} \right) \\
	&=k_{-1},
	\end{align*}
	where the first inequality comes from Definition \ref{def2}. Hence, by Definition \ref{def4},
	\begin{equation*}
		\delta^{\elp}_{k_{-1}}(\mathcal C_{1,\mathbf v})
		\le|\supp(\mathcal C_{1,\mathbf v})|
		\le n_F-|\mathcal U_1|.
	\end{equation*}
	Combining this with $|\mathcal U_1|\le n_{I_1}$ proves the result.
\end{IEEEproof}

\begin{remark}
	In \cite[Proposition 3]{gruica2026linear}, the number of unchanged symbols can further be bounded by $|\mathcal U_i| \le \min \{ n_{I_i}, n_F-d_{k_{-i}}(\mathcal C_F) \}$ for linear codes, since $d_{k_{-i}}(\mathcal C_F)\le d_{k_{-i}}(\mathcal C_{i,\mathbf v})$. This comparison need not hold for nonlinear codes because $\delta^{\elp}_{\rho}(\mathcal C)\le \delta^{\elp}_{\rho}(\mathcal C')$ is false in general for $\mathcal C'\subseteq\mathcal C$, as the following example shows.
\end{remark}

\begin{example}\label{exa1}
	Let $\mathcal C=\{00,01,10\}\subseteq\mathbb F_2^2$. For $J=\{1\}$,
	\begin{align*}
		&H(X_J|X_{[2]\setminus J}) \\ 
		&= \frac{2}{3} H(X_J|X_{[2]\setminus J}=0) + \frac{1}{3} H(X_J|X_{[2]\setminus J}=1) \\
		&= \frac{2}{3} \left( -\frac{1}{2} \log_2 \frac{1}{2} -\frac{1}{2} \log_2 \frac{1}{2} \right) + \frac{1}{3} \left( -1 \log_2 1 \right) \\
		&=\frac{2}{3}.
	\end{align*}
	The same calculation gives $H(X_J|X_{[2]\setminus J})=\frac{2}{3}$ for $J=\{2\}$. Hence $h_1(\mathcal C)=\frac{2}{3}$, $h_2(\mathcal C)=\log_2 3$, and Definition \ref{def4} gives $\delta^{\elp}_1(\mathcal C)=2$. For the subset $\mathcal C'=\{00,01\}$, both ELP values equal one, so $\delta^{\elp}_1(\mathcal C')=1$. Consequently, $\delta^{\elp}_1(\mathcal C)>\delta^{\elp}_1(\mathcal C')$.
\end{example}

Although $\delta^{\elp}_{\rho}(\mathcal C) \le \delta^{\elp}_{\rho}(\mathcal C')$ does not hold in general for $\mathcal C'\subseteq\mathcal C$, we can obtain a bound directly from the ELP of the final code by using conditional entropy.

\begin{lemma}\label{lem7}
    Let $(\mathcal C_{\mathbf I},\mathcal C_F,\sigma)$ be a convertible code. For every $i\in[t]$,
    \begin{equation*}
        |\mathcal U_i|
        \le
        \min\left\{
        n_{I_i},
        n_F-\delta^{\elp}_{k_{-i}}(\mathcal C_F)
        \right\}.
    \end{equation*}
\end{lemma}

\begin{IEEEproof}
    Without loss of generality, let $i=1$. Fix any $\mathbf u\in\mathcal C_F|_{\mathcal U_1}$. By the definition of unchanged coordinates, there is an initial codeword $\mathbf v\in\mathcal C_{I_1}$ such that
    	\begin{equation*}
        	\mathbf v|_{\mathcal U_1}=\mathbf u.
    	\end{equation*}
    	Keeping $\mathbf v$ fixed and varying the other initial codewords produces
    	\begin{equation*}
        	\prod_{j=2}^{t}M_{I_j} =q^{k_{-1}}
    	\end{equation*}
    distinct final codewords, all of which have projection $\mathbf u$ on $\mathcal U_1$. Since $\sigma$ is bijective, these final codewords have distinct projections on $\Omega_F\setminus\mathcal U_1$. Therefore,
    \begin{equation*}
        \left| \left\{ \mathbf c|_{\Omega_F\setminus\mathcal U_1}: \mathbf c\in\mathcal C_F,\  \mathbf c|_{\mathcal U_1}=\mathbf u \right\} \right| \ge q^{k_{-1}}.
    \end{equation*}
    Define
    \begin{equation*}
        \mathcal C_F(\mathbf u)
        =\{\mathbf c\in\mathcal C_F:\mathbf c|_{\mathcal U_1}=\mathbf u\}.
    \end{equation*}
    For every $\mathbf c\in\mathcal C_F(\mathbf u)$, the event $\{X^{(\mathcal C_F)}=\mathbf c\}$ implies $\{X_{\mathcal U_1}^{(\mathcal C_F)}=\mathbf u\}$. Uniformity on $\mathcal C_F$ therefore gives
    \begin{align*}
        &\Pr\!\left(X^{(\mathcal C_F)}=\mathbf c
        \middle|X_{\mathcal U_1}^{(\mathcal C_F)}=\mathbf u\right)\\
        &= \frac{\Pr\!\left(X^{(\mathcal C_F)}=\mathbf c, X_{\mathcal U_1}^{(\mathcal C_F)}=\mathbf u\right)}{\Pr\!\left(X_{\mathcal U_1}^{(\mathcal C_F)}=\mathbf u\right)} \\
        &=\frac{1/|\mathcal C_F|}{|\mathcal C_F(\mathbf u)|/|\mathcal C_F|} \\
        &=\frac{1}{|\mathcal C_F(\mathbf u)|}.
    \end{align*}
    Hence the conditional distribution is uniform on $\mathcal C_F(\mathbf u)$. Since all codewords in $\mathcal C_F(\mathbf u)$ agree on $\mathcal U_1$, distinct codewords have distinct projections on $\Omega_F\setminus\mathcal U_1$. The corresponding conditional projection is therefore uniform on a set of at least $q^{k_{-1}}$ elements. Thus, for every $\mathbf u\in\mathcal C_F|_{\mathcal U_1}$,
    \begin{equation}\label{eqn3}
        H\left( X_{\Omega_F\setminus\mathcal U_1}^{(\mathcal C_F)} \,\middle|\, X_{\mathcal U_1}^{(\mathcal C_F)}=\mathbf u \right) \ge k_{-1}.
    \end{equation}

    Averaging \eqref{eqn3} over all possible values of $X_{\mathcal U_1}^{(\mathcal C_F)}$ gives
    \begin{align*}
        &h_{n_F-|\mathcal U_1|}(\mathcal C_F) \\
        &\ge H\left( X_{\Omega_F\setminus\mathcal U_1}^{(\mathcal C_F)} \,\middle|\, X_{\mathcal U_1}^{(\mathcal C_F)} \right)\\
        &= \sum_{\mathbf u\in\mathcal C_F|_{\mathcal U_1}} \Pr\left( X_{\mathcal U_1}^{(\mathcal C_F)}=\mathbf u \right) H\left( X_{\Omega_F\setminus\mathcal U_1}^{(\mathcal C_F)} \,\middle|\, X_{\mathcal U_1}^{(\mathcal C_F)}=\mathbf u \right)\\
        &\ge \sum_{\mathbf u\in\mathcal C_F|_{\mathcal U_1}} \Pr\left( X_{\mathcal U_1}^{(\mathcal C_F)}=\mathbf u \right) \cdot k_{-1}\\
        &=k_{-1}.
    \end{align*}
    It follows from Definition \ref{def4} that
    \begin{equation*}
        \delta^{\elp}_{k_{-1}}(\mathcal C_F) \le n_F-|\mathcal U_1|.
    \end{equation*}
    Thus,
    \begin{equation*}
        |\mathcal U_1| \le n_F-\delta^{\elp}_{k_{-1}}(\mathcal C_F).
    \end{equation*}
    Combining this inequality with the trivial bound $|\mathcal U_1|\le n_{I_1}$ proves the result.
\end{IEEEproof}

\begin{remark}
Lemma \ref{lem6} uses the ELP of the code obtained by fixing one initial codeword, whereas Lemma \ref{lem7} uses the ELP of the full final code $\mathcal C_F$. In the linear setting, the two bounds can be compared using the monotonicity of GHW. This monotonicity may fail for nonlinear codes, as shown by Example \ref{exa1}.
\end{remark}

\subsection{Joint Bounds}

The corresponding individual bound in \cite[Proposition 3]{gruica2026linear} follows from Lemma \ref{lem7} by specializing to linear codes. However, applying the individual bound separately to each initial code does not fully describe how several initial codes share the coordinates of the final code. We therefore extend Lemma \ref{lem7} and derive a joint upper bound for every selected collection of initial codes, which provides an additional restriction that may be stronger than the sum of the individual bounds.

\begin{lemma}\label{lem8}
	Let $(\mathcal C_{\mathbf{I}},\mathcal C_F,\sigma)$ be a convertible code. For every nonempty $I\subseteq[t]$,
	\begin{equation*}
		\sum_{j\in I}|\mathcal U_j|\le \min\left\{ \sum_{j\in I}n_{I_j}, n_F-\delta^{\elp}_{k_{-I}}(\mathcal C_F) \right\}.
	\end{equation*}
\end{lemma}

The proof of Lemma \ref{lem8} is similar to that of Lemma \ref{lem7}; see Appendix \ref{app1} for details.

\begin{remark}\label{rem1}
	Let $t=3$ and consider $|\mathcal U_1|+|\mathcal U_2|$. Applying Lemma \ref{lem7} separately to the first two initial codes gives
	\begin{equation*}
		|\mathcal U_1|+|\mathcal U_2|\le \min\{n_{I_1}+n_{I_2},\alpha\},
	\end{equation*}
	where
	\begin{equation*}
		\alpha=2n_F-\delta^{\elp}_{k_{-1}}(\mathcal C_F)-\delta^{\elp}_{k_{-2}}(\mathcal C_F).
	\end{equation*}
	Applying Lemma \ref{lem8} with $I=\{1,2\}$ gives
	\begin{equation*}
		|\mathcal U_1|+|\mathcal U_2|\le \min\{n_{I_1}+n_{I_2},\beta\},
	\end{equation*}
	where
	\begin{equation*}
		\beta=n_F-\delta^{\elp}_{k_{-\{1,2\}}}(\mathcal C_F).
	\end{equation*}
	Appendices \ref{app2}--\ref{app4} show that all three relations are possible: $\alpha=\beta$, $\alpha>\beta$, and $\alpha<\beta$, respectively. Hence neither bound is always stronger, and together they give
	\begin{equation*}
		|\mathcal U_1|+|\mathcal U_2|\le \min\{n_{I_1}+n_{I_2},\alpha,\beta\}.
	\end{equation*}
\end{remark}

\subsection{Consequences for Minimum Distance with Noninteger Dimensions}

We next use the ELP bounds to obtain bounds involving minimum distance. This loses part of the ELP information but gives simpler expressions involving familiar code parameters. In contrast to the linear case, $k_i$ and $k_{-i}$ need not be integers.

\begin{corollary}\label{cor1}
	Let $(\mathcal C_{\mathbf I},\mathcal C_F,\sigma)$ be a convertible code. Then
	\begin{align}
	 |\mathcal U_i|
	 &\le \min\left\{n_{I_i}, n_F-\max_{\mathbf v\in\mathcal C_{I_i}} d(\mathcal C_{i,\mathbf v})-\lceil k_{-i}\rceil+1\right\}, \label{eqn4}\\
	 |\mathcal U_i|
	 &\le \min\left\{n_{I_i}, n_F-d(\mathcal C_F)-\lceil k_{-i}\rceil+1\right\}. \label{eqn5}
	\end{align}
\end{corollary}

\begin{IEEEproof}
	Lemma \ref{lem4}, applied to each $\mathcal C_{i,\mathbf v}$, gives
	\begin{equation*}
	 \delta_{k_{-i}}^{\elp}(\mathcal C_{i,\mathbf v}) \ge d(\mathcal C_{i,\mathbf v})+\lceil k_{-i}\rceil-1.
	\end{equation*}
	Substitution into Lemma \ref{lem6}, followed by minimization over $\mathbf v$, proves \eqref{eqn4}. Applying the same ELP inequality to $\mathcal C_F$ in Lemma \ref{lem7} proves \eqref{eqn5}.
\end{IEEEproof}

The bound in \eqref{eqn4} can be stronger because the sets $\mathcal C_{i,\mathbf v}$ may have larger minimum distance than $\mathcal C_F$. 
We next give a total bound that preserves the integer contribution of every initial code separately.

\begin{lemma}\label{lem9}
	Let $(\mathcal C_{\mathbf I},\mathcal C_F,\sigma)$ be a convertible code with $t\ge2$. Then
	\begin{equation}\label{eqn6}
	 |\mathcal U| \le \left\lfloor \frac{t n_F-t\bigl(d(\mathcal C_F)-1\bigr) -\sum_{i=1}^{t}\lceil k_i\rceil}{t-1} \right\rfloor .
	\end{equation}
\end{lemma}

\begin{IEEEproof}
	Fix $i\in[t]$ and $\mathbf v_{-i}\in\prod_{j\ne i}\mathcal C_{I_j}$, and let
	\begin{equation*}
		\mathcal C_{i\mid\mathbf v_{-i}}
		=\{\sigma(\mathbf v_1,\ldots,\mathbf v_{i-1},
		\mathbf c_i,\mathbf v_{i+1},\ldots,\mathbf v_t):
		\mathbf c_i\in\mathcal C_{I_i}\}.
	\end{equation*}
	The conversion is bijective, so $|\mathcal C_{i\mid\mathbf v_{-i}}|=M_{I_i}$. Only the coordinates in $\mathcal U_i\sqcup\mathcal W$ can vary on this set, and hence
	\begin{equation*}
	 	|\supp(\mathcal C_{i\mid\mathbf v_{-i}})|\le|\mathcal U_i|+|\mathcal W|.
	\end{equation*}
	Since $\mathcal C_{i\mid\mathbf v_{-i}}\subseteq\mathcal C_F$, its minimum distance is at least $d(\mathcal C_F)$. The Singleton bound gives
	\begin{equation*}
	 	M_{I_i}\le q^{|\supp(\mathcal C_{i\mid\mathbf v_{-i}})|-d(\mathcal C_F)+1}.
	\end{equation*}
	Since $|\supp(\mathcal C_{i\mid\mathbf v_{-i}})|$ is an integer, we have
	\begin{equation*}
	 	|\mathcal U_i|+|\mathcal W| \ge d(\mathcal C_F)+\lceil k_i\rceil-1.
	\end{equation*}
	Summing over $i\in[t]$ and using $|\mathcal W|=n_F-|\mathcal U|$ yields
	\begin{equation*}
	 	(t-1)|\mathcal U| \le t n_F-t\bigl(d(\mathcal C_F)-1\bigr) -\sum_{i=1}^{t}\lceil k_i\rceil.
	\end{equation*}
	Taking the floor proves \eqref{eqn6}.
\end{IEEEproof}

\begin{remark}\label{rem2}
	If every code is linear, then $k_i$ and $k_{-i}$ are integers, and the sets $\mathcal C_{i,\mathbf v}$ are cosets of linear subcodes. The ceiling operations disappear, recovering the minimum distance bound in \cite[Proposition 7]{gruica2026linear}. For nonlinear codes, the dimensions may be nonintegral and the sets $\mathcal C_{i,\mathbf v}$ need not be cosets or have a common minimum distance. Appendix \ref{app5} gives a nonlinear example in which the bound in Lemma \ref{lem9} is attained and separate rounding strictly decreases the right-hand side of the corresponding unrounded expression.
\end{remark}

We now turn from the number of symbols kept unchanged to the number that must be read.

\section{Bounds on Read Symbols}\label{sec4}

\subsection{A General Bound for Real Dimensions}

We derive lower bounds on the number of symbols read from each initial code. The dimensions $k_i$ defined in Section~\ref{sec2} may be nonintegral.
\begin{lemma}\label{lem10}
	Let $(\mathcal C_{\mathbf{I}},\mathcal C_F,\sigma)$ be a convertible code. For every $i\in[t]$,
	\begin{equation}\label{eqn7}
	|\mathcal R_i| \ge \max\left\{ 0,\, \left\lceil k_i-|\mathcal U_i|\right\rceil,\, \max_{0<r\le k_i}g_i(r) \right\},
	\end{equation}
	where $g_i(r)=\lfloor k_i-r- ( |\mathcal U_i|-\delta_r^{\elp}(\mathcal C_F)+1 )^+ \rfloor+1$ and $(a)^+=\max\{0,a\}$. 
\end{lemma}

\begin{IEEEproof}
	Without loss of generality, let $i=1$. For each
	\begin{equation*}
		\mathbf v_{-1}=(\mathbf v_2,\ldots,\mathbf v_t)
		\in\prod_{j=2}^t\mathcal C_{I_j},
	\end{equation*}
	define
	\begin{equation*}
		\mathcal C_{1\mid\mathbf v_{-1}}
		=\{\sigma(\mathbf c_1,\mathbf v_2,\ldots,\mathbf v_t):
		\mathbf c_1\in\mathcal C_{I_1}\}.
	\end{equation*}
	The bijectivity of $\sigma$ gives $|\mathcal C_{1\mid\mathbf v_{-1}}|=M_{I_1}$.

	We first prove the bound corresponding to the zero entropy level. With the other initial codewords fixed, the final projection on $\mathcal U_1\sqcup\mathcal W$ is determined by the $|\mathcal U_1|$ unchanged symbols and the $|\mathcal R_1|$ symbols read from the first initial codeword. Since the restriction of $\sigma$ to $\mathcal C_{I_1}$ is injective,
	\begin{equation*}
		M_{I_1}\le q^{|\mathcal U_1|+|\mathcal R_1|}.
	\end{equation*}
	Thus, because $|\mathcal R_1|$ is an integer,
	\begin{equation}\label{eqn8}
		|\mathcal R_1| \ge\left\lceil k_1-|\mathcal U_1|\right\rceil.
	\end{equation}

	We next prove the bound corresponding to a positive entropy level. Fix $r\in(0,k_1]$ and write
	\begin{equation*}
		\Delta_r = \bigl(|\mathcal U_1|-\delta_r^{\elp}(\mathcal C_F)+1\bigr)^+.
	\end{equation*}
	Choose $T\subseteq\mathcal U_1$ such that
	\begin{equation}\label{eqn_T}
		|T| = \min\left\{|\mathcal U_1|, \delta_r^{\elp}(\mathcal C_F)-1\right\}.
	\end{equation}
	Since $h_0(\mathcal C_F)=0<r\le h_{n_F}(\mathcal C_F)$, the set in the definition of $\delta_r^{\elp}(\mathcal C_F)$ is nonempty. Lemma \ref{lem1} further gives
	\begin{equation*}
		\delta_r^{\elp}(\mathcal C_F)\ge\lceil r\rceil\ge1.
	\end{equation*}
	Hence this choice of $T$ is well defined, i.e., $|T|\ge 0$. Moreover,
	\begin{equation*}
		|\mathcal U_1\setminus T|=\Delta_r.
	\end{equation*}
	
	Set $S=\Omega_F\setminus T$ and let $X^{(\mathcal C_{1\mid\mathbf v_{-1}})}$ be uniformly distributed over $\mathcal C_{1\mid\mathbf v_{-1}}$. Once $\mathbf v_{-1}$ is fixed, the coordinates in $S$ are determined by the symbols in $\mathcal U_1\setminus T$ and the symbols read from $\mathcal R_1$. Therefore,
	\begin{equation*}
		\log_q\left| \mathcal C_{1\mid\mathbf v_{-1}}|_S \right| \le |\mathcal R_1|+\Delta_r.
	\end{equation*}
	Because $T$ and $S$ partition $\Omega_F$ and $H(X^{(\mathcal C_{1\mid\mathbf v_{-1}})})=k_1$, we have
	\begin{equation}\label{eqn9}
	\begin{split}
		&H\left( X_T^{(\mathcal C_{1\mid\mathbf v_{-1}})} \middle| X_S^{(\mathcal C_{1\mid\mathbf v_{-1}})} \right)\\
		&= k_1-H\left(X_S^{(\mathcal C_{1\mid\mathbf v_{-1}})}\right)\\ 
		&\ge k_1-\log_q\left| \mathcal C_{1\mid\mathbf v_{-1}}|_S \right|\\
		&\ge k_1-|\mathcal R_1|-\Delta_r.
	\end{split}
	\end{equation}

	Next, let $X^{(\mathcal C_F)}$ be uniformly distributed over $\mathcal C_F$, and define $V$ on this probability space. For an outcome $X^{(\mathcal C_F)}=\mathbf c$, let
	\begin{equation*}
		(\mathbf c_{I_1},\ldots,\mathbf c_{I_t}) =\sigma^{-1}(\mathbf c)
	\end{equation*}
	and set
	\begin{equation*}
		V=(\mathbf c_{I_2},\ldots,\mathbf c_{I_t}).
	\end{equation*}
	The event $\{V=\mathbf v_{-1}\}$ is exactly $\{X^{(\mathcal C_F)}\in\mathcal C_{1\mid\mathbf v_{-1}}\}$. These sets $\{X^{(\mathcal C_F)}\in\mathcal C_{1\mid\mathbf v_{-1}}\}, \mathbf v_{-1}\in\prod_{j=2}^t\mathcal C_{I_j}$ are pairwise disjoint, each has size $M_{I_1}$, and their union is $\mathcal C_F$. It follows that
	\begin{equation*}
		\Pr(V=\mathbf v_{-1}) = \frac{M_{I_1}}{M_F} = \frac{1}{\prod_{j=2}^tM_{I_j}}.
	\end{equation*}
	In particular, $V$ is uniform. For every $\mathbf c\in\mathcal C_{1\mid\mathbf v_{-1}}$, the event $\{X^{(\mathcal C_F)}=\mathbf c\}$ implies $\{V=\mathbf v_{-1}\}$. Hence
	\begin{equation*}
		\Pr\!\left(X^{(\mathcal C_F)}=\mathbf c \middle|V=\mathbf v_{-1}\right)
		=\frac{1/M_F}{M_{I_1}/M_F}
		=\frac{1}{M_{I_1}}.
	\end{equation*}
	The conditional distribution of $X^{(\mathcal C_F)}$ is therefore uniform on $\mathcal C_{1\mid\mathbf v_{-1}}$. This is the same distribution used to define $X^{(\mathcal C_{1\mid\mathbf v_{-1}})}$, although the two random variables were introduced on different probability spaces. Consequently,
	\begin{equation}\label{eqn10}
	H\left( X_T^{(\mathcal C_F)} \middle|
	X_S^{(\mathcal C_F)},V=\mathbf v_{-1} \right)
	=H\left( X_T^{(\mathcal C_{1\mid\mathbf v_{-1}})}
	\middle| X_S^{(\mathcal C_{1\mid\mathbf v_{-1}})} \right).
	\end{equation}

	We can now prove the bound. Suppose for contradiction that 
	\begin{equation*}
		|\mathcal R_1| \le\left\lfloor k_1-r-\Delta_r\right\rfloor = g_1(r)-1.
	\end{equation*}
	Then $|\mathcal R_1|\le k_1-r-\Delta_r$. Since \eqref{eqn9} holds for every $\mathbf v_{-1}\in\prod_{j=2}^t\mathcal C_{I_j}$,
	\begin{equation*}
		H\left( X_T^{(\mathcal C_{1\mid\mathbf v_{-1}})} \middle| X_S^{(\mathcal C_{1\mid\mathbf v_{-1}})} \right)\ge r
	\end{equation*}
	Therefore, conditioning reduces entropy and \eqref{eqn10} gives
	\begin{align*}
		&H\left( X_T^{(\mathcal C_F)} \middle| X_S^{(\mathcal C_F)} \right)\\
		&\ge H\left( X_T^{(\mathcal C_F)} \middle| X_S^{(\mathcal C_F)},V \right)\\
		&= \sum_{\mathbf v_{-1}} \Pr(V=\mathbf v_{-1}) H\left( X_T^{(\mathcal C_F)} \middle| X_S^{(\mathcal C_F)},V=\mathbf v_{-1} \right)\\
		&= \sum_{\mathbf v_{-1}} \Pr(V=\mathbf v_{-1}) H\left( X_T^{(\mathcal C_{1\mid\mathbf v_{-1}})} \middle| X_S^{(\mathcal C_{1\mid\mathbf v_{-1}})} \right)\\
		&\ge \sum_{\mathbf v_{-1}}\Pr(V=\mathbf v_{-1}) \cdot r \\
		&=r,
	\end{align*}
	where the sums are over $\mathbf v_{-1}\in\prod_{j=2}^t\mathcal C_{I_j}$. By the definition of the ELP,
	\begin{equation*}
		h_{|T|}(\mathcal C_F) \ge H\left( X_T^{(\mathcal C_F)} \middle| X_S^{(\mathcal C_F)} \right) \ge r.
	\end{equation*}
	Hence $|T|\ge\delta_r^{\elp}(\mathcal C_F)$, contradicting \eqref{eqn_T}, i.e.,
	\begin{equation*}
		|T| \le\delta_r^{\elp}(\mathcal C_F)-1.
	\end{equation*}
	We conclude that
	\begin{equation*}
		|\mathcal R_1| \ge\left\lfloor k_1-r-\Delta_r\right\rfloor+1 =g_1(r)
	\end{equation*}
	for every $0<r\le k_1$. Combining this inequality with \eqref{eqn8} and $|\mathcal R_1|\ge0$, and then taking the maximum over $r$, proves \eqref{eqn7}.
\end{IEEEproof}

\subsection{A Threshold Form}

We now present a simplified lower bound that is often easier to apply. It replaces the optimization over entropy levels by a single threshold determined by the unchanged coordinates.

\begin{theorem}\label{thm1}
	Let $(\mathcal C_{\mathbf{I}},\mathcal C_F,\sigma)$ be a convertible code and fix $i\in[t]$. Define
	\begin{equation*}
		\theta_i=\sup\{r\in[0,k_i]:\delta^{\elp}_r(\mathcal C_F)\le|\mathcal U_i|\}.
	\end{equation*}
	Then the right side of \eqref{eqn7} equals
	$\lceil k_i-\theta_i\rceil$. In particular,
	\begin{equation*}
		|\mathcal R_i|\ge \left\lceil k_i-\theta_i\right\rceil
		=\max\{0,\lceil k_i-h_{|\mathcal U_i|}(\mathcal C_F)\rceil \}.
	\end{equation*}
\end{theorem}

\begin{IEEEproof}
	For every $r\in(0,k_i]$, Definition \ref{def4} and Lemma \ref{lem2} show 
	\begin{equation*}
		|\mathcal U_i| \ge \delta_r^{\elp}(\mathcal C_F)
		= \min\{j\in[n_F]: h_j(\mathcal C_F)\ge r\}
	\end{equation*}
	if and only if
	\begin{equation*}
		r\le h_{|\mathcal U_i|}(\mathcal C_F).
	\end{equation*}
	Since $\delta_0^{\elp}(\mathcal C_F)=0$, it follows that
	\begin{equation}\label{eqn11}
		\theta_i
		=\min\left\{k_i,h_{|\mathcal U_i|}(\mathcal C_F)\right\}.
	\end{equation}
	In particular, Lemma \ref{lem2} and
	$h_0(\mathcal C_F)=0$ imply
	\begin{equation}\label{eqn12}
		\theta_i
		\le h_{|\mathcal U_i|}(\mathcal C_F)
		\le|\mathcal U_i|.
	\end{equation}

	We now consider every entropy level in Lemma \ref{lem10}. 
	
	First, let $0<r\le\theta_i$. Then $\delta_r^{\elp}(\mathcal C_F)\le|\mathcal U_i|$ by the definition of $\theta_i$. Let $d=\delta_r^{\elp}(\mathcal C_F)$. Since $r>0$ and $h_0(\mathcal C_F)=0$, we have $d\ge1$. Otherwise, $d=0$ implies a contradiction $r=0$. By Definition \ref{def4}, $d$ is the smallest index such that $h_d(\mathcal C_F)\ge r$. Therefore, every index $j$ smaller than $d$ satisfies $h_j(\mathcal C_F)<r$. In particular,
	\begin{equation*}
		h_{d-1}(\mathcal C_F)<r.
	\end{equation*}
	By Lemma \ref{lem2},
	\begin{equation*}
		\begin{split}
		h_{|\mathcal U_i|}(\mathcal C_F)
		&\le h_{d-1}(\mathcal C_F) +|\mathcal U_i|-d+1\\
		&<r+|\mathcal U_i|-d+1.
		\end{split}
	\end{equation*}
	Together with \eqref{eqn12}, this yields
	\begin{equation*}
		|\mathcal U_i|-d+1
		>\theta_i-r.
	\end{equation*}
	Since $d=\delta_r^{\elp}(\mathcal C_F)\le|\mathcal U_i|$, the term involving $(\cdot)^+$ in \eqref{eqn7} equals $|\mathcal U_i|-d+1$.
	For any real numbers $a<b$, we have $\lfloor a\rfloor+1\le\lceil b\rceil$. Hence
	\begin{equation}\label{eqn13}
	\begin{split}
		g_i(r)
		&=\left\lfloor k_i-r-\left(|\mathcal U_i|-d+1\right)\right\rfloor+1\\
		&\le\left\lceil k_i-r-(\theta_i-r)\right\rceil\\
		&=\left\lceil k_i-\theta_i\right\rceil.
	\end{split}
	\end{equation}

	Second, let $\theta_i<r\le k_i$. Then
	$\delta_r^{\elp}(\mathcal C_F)>|\mathcal U_i|$, so the term involving $(\cdot)^+$ in \eqref{eqn7} is zero. Therefore,
	\begin{equation}\label{eqn14}
		g_i(r)=\lfloor k_i-r\rfloor+1 \le\left\lceil k_i-\theta_i\right\rceil.
	\end{equation}

	Third, \eqref{eqn12} gives
	\begin{equation}\label{eqn15}
		\left\lceil k_i-|\mathcal U_i|\right\rceil
		\le\left\lceil k_i-\theta_i\right\rceil.
	\end{equation}
	
	If $\theta_i<k_i$, equations \eqref{eqn13}, \eqref{eqn14}, and \eqref{eqn15} show that the right side of \eqref{eqn7} equals $\lceil k_i-\theta_i\rceil$ because $r$ may be chosen arbitrarily close to $\theta_i$ from above. If $\theta_i=k_i$, all positive entropy levels lie in the range $0<r\le\theta_i$, \eqref{eqn13} gives $g_i(r)\le0$, and \eqref{eqn15} gives $\lceil k_i-|\mathcal U_i|\rceil\le0$. Thus the zero term in \eqref{eqn7} is maximal, again with value
	\begin{equation*}
		0=\left\lceil k_i-\theta_i\right\rceil.
	\end{equation*}
	This proves the claimed equality and the resulting read bound.
\end{IEEEproof}

Theorem \ref{thm1} evaluates the full optimization in Lemma \ref{lem10}. The quantity $h_{|\mathcal U_i|}(\mathcal C_F)$ is the maximum residual entropy that can be supported on a set of $|\mathcal U_i|$ final coordinates when the remaining coordinates are known. Thus, $\theta_i$ is this quantity capped at $k_i$, and the proof accounts for entropy levels both below and above this threshold.

The following direct consequence records a simple necessary condition for a nontrivial convertible code.

\begin{corollary}\label{cor2}
	Let $(\mathcal C_{\mathbf I},\mathcal C_F,\sigma)$ be a nontrivial convertible code. Then, there exists $i\in[t]$ such that
	\begin{equation}\label{eqn16}
		|\mathcal U_i|\ge d(\mathcal C_F).
	\end{equation}
\end{corollary}

\begin{IEEEproof}
	Suppose that $|\mathcal U_i|<d(\mathcal C_F)$ for every $i\in[t]$. Lemma \ref{lem4} gives
	\begin{equation*}
		\delta_r^{\elp}(\mathcal C_F)\ge d(\mathcal C_F) + \lceil r \rceil-1 > |\mathcal U_i|
	\end{equation*}
	for every $0<r\le k_i$. Hence $\theta_i=0$, and Theorem \ref{thm1} gives $|\mathcal R_i|\ge\lceil k_i\rceil$ for $i\in[t]$. Therefore,
	\begin{equation}
		\sum_{i=1}^{t}|\mathcal R_i| \ge \sum_{i=1}^{t} \lceil k_i\rceil \ge k_F,
	\end{equation}
	which contradicts the nontriviality assumption in Definition \ref{def5}.
\end{IEEEproof}


\section{Comparison with the Linear Theory}\label{sec5}

We now specialize the preceding bounds to linear merge conversions and compare them with \cite{gruica2026linear}. This specialization recovers the general bounds on unchanged and read symbols therein. It also makes clear what changes in the nonlinear setting: entropy levels may be real, rounding becomes necessary, and the proofs cannot rely on linear subcodes or duality. Table \ref{tab1} summarizes the comparison.

\begin{proposition}\label{prop1}
	Let $\mathcal C\subseteq\mathbb F_q^n$ be a linear $[n,k]_q$ code, and let $d_j(\mathcal C)=\min\{|\supp(\mathcal D)| : \mathcal D \subseteq \mathcal C, \dim(\mathcal D)=j\}$ denote its $j$th GHW. Then, for every real $\rho\in(0,k]$,
	\begin{equation}\label{eqn17}
		\delta_\rho^{\elp}(\mathcal C) =d_{\lceil\rho\rceil}(\mathcal C).
	\end{equation}
\end{proposition}

\begin{IEEEproof}
	For $J\subseteq[n]$, define the shortened subcode supported in $J$ by
	\begin{equation*}
		\mathcal C(J)
		= \{\mathbf c\in\mathcal C: \mathbf c|_{J^c}=\mathbf0\} \subseteq \mathcal C.
	\end{equation*}
	Fix any $\mathbf y\in\mathcal C|_{J^c}$ and choose $\mathbf c_0\in\mathcal C$ such that $\mathbf c_0|_{J^c}=\mathbf y$. Since $\mathcal C$ is linear, the set of codewords whose projection onto $J^c$ equals $\mathbf y$ is
	\begin{equation*}
		\left\{ \mathbf c\in\mathcal C: \mathbf c|_{J^c}=\mathbf y \right\} = \mathbf c_0+\mathcal C(J).
	\end{equation*}
	For every $\mathbf c\in\mathbf c_0+\mathcal C(J)$, uniformity on $\mathcal C$ gives
	\begin{equation*}
		\Pr\!\left(X^{(\mathcal C)}=\mathbf c
		\middle|X_{J^c}^{(\mathcal C)}=\mathbf y\right)
		=\frac{1/|\mathcal C|}{|\mathcal C(J)|/|\mathcal C|}
		=\frac{1}{|\mathcal C(J)|}.
	\end{equation*}
	Therefore, conditioned on $X_{J^c}^{(\mathcal C)}=\mathbf y$, the random vector $X_J^{(\mathcal C)}$ is uniformly distributed over
	\begin{equation*}
		(\mathbf c_0+\mathcal C(J))|_J,
	\end{equation*}
	which has $q^{\dim\mathcal C(J)}$ elements. Hence
	\begin{equation*}
		H\left( X_J^{(\mathcal C)} \middle| X_{J^c}^{(\mathcal C)}=\mathbf y \right)
		= \dim\mathcal C(J).
	\end{equation*}
	Since this value is the same for every $\mathbf y\in\mathcal C|_{J^c}$, we obtain
	\begin{equation*}
		H\left( X_J^{(\mathcal C)} \middle| X_{J^c}^{(\mathcal C)} \right)
		= \dim\mathcal C(J).
	\end{equation*}
	It follows that
	\begin{equation}\label{eqn18}
		h_i(\mathcal C)
		=\max_{\substack{J\subseteq[n]\\|J|=i}}\dim\mathcal C(J).
	\end{equation}
	
	For every integer $j\in[k]$, \eqref{eqn18} shows that $h_i(\mathcal C)\ge j$ if and only if there exists $J\subseteq[n], |J|=i$ such that
	\begin{equation*}
		\dim\mathcal C(J)\ge j.
	\end{equation*}
	
	We claim that there exists $J\subseteq[n], |J|=i$ such that $\dim\mathcal C(J)\ge j$ if and only if $d_j(\mathcal C)\le i$. 
	
	If $\dim\mathcal C(J)\ge j$, we can choose a $j$-dimensional subcode $\mathcal D\subseteq\mathcal C(J)$. Since every codeword in $\mathcal C(J)$ is zero outside $J$, we have
	\begin{equation*}
    	\supp(\mathcal D)
    	\subseteq \supp(\mathcal C(J))
    	\subseteq J.
	\end{equation*}
	Therefore,
	\begin{align*}
    	d_j(\mathcal C)
    	&=\min_{\substack{\mathcal E\subseteq\mathcal C\\ \dim(\mathcal E)=j}} |\supp(\mathcal E)|\\
    	&\le |\supp(\mathcal D)|\\
    	&\le |\supp(\mathcal C(J))|\\
    	&\le |J|\\
    	&=i.
	\end{align*}
	Therefore, if $\dim\mathcal C(J)\ge j$ then $d_j(\mathcal C)\le i$.
	
	Conversely, suppose that $d_j(\mathcal C)\le i$. Then there exists a $j$-dimensional subcode $\mathcal D\subseteq\mathcal C$ such that
	\begin{equation*}
    	|\supp(\mathcal D)|=d_j(\mathcal C)\le i.
	\end{equation*}
	Choose a set $J\subseteq[n]$ with $|J|=i$ and $\supp(\mathcal D)\subseteq J$. Every codeword in $\mathcal D$ is zero outside $J$, and hence
	\begin{equation*}
    	\mathcal D\subseteq\mathcal C(J).
	\end{equation*}
	It follows that
	\begin{equation*}
    	\dim\mathcal C(J)\ge\dim\mathcal D=j.
	\end{equation*}
	Therefore, if $d_j(\mathcal C)\le i$ then $\dim\mathcal C(J)\ge j$.
	
	The preceding equivalences imply
	\begin{equation*}
		h_i(\mathcal C)\ge j \quad\Longleftrightarrow\quad d_j(\mathcal C)\le i.
	\end{equation*}
	Hence,
	\begin{equation}\label{eqn19}
		\delta_j^{\elp}(\mathcal C)
		=\min\{i\in[n]:h_i(\mathcal C)\ge j\}
		=d_j(\mathcal C).
	\end{equation}
	Finally, $h_i(\mathcal C)$ is integer valued for linear codes. Thus, for every $\rho\in(0,k]$,
	\begin{equation*}
		h_i(\mathcal C)\ge\rho
		\quad\Longleftrightarrow\quad
		h_i(\mathcal C)\ge\lceil\rho\rceil.
	\end{equation*}
	Combining this equivalence with \eqref{eqn19} yields $\delta_\rho^{\elp}(\mathcal C) = d_{\lceil\rho\rceil}(\mathcal C)$.
\end{IEEEproof}

Combining Proposition \ref{prop1} with Lemma \ref{lem5} gives, for integers $1\le j_1\le j_2\le k$,
\begin{equation*}
	d_{j_2}(\mathcal C)
	\ge d_{j_1}(\mathcal C)+(j_2-j_1),
\end{equation*}
which is the standard strict growth property of GHW for linear codes.

\subsection{Bounds on Unchanged Symbols}

Suppose all initial and final codes are linear and $\sigma$ is linear. Then the quantities defined in Section~\ref{sec2} satisfy $k_i=\dim(\mathcal C_{I_i})$ and $k_{-i}=\sum_{j\ne i}k_j$. Recall 
\begin{equation*}
	\mathcal C_{i,\mathbf{0}}
	=
	\sigma\left(
	\prod_{j=1}^{i-1}\mathcal C_{I_j}
	\times\{\mathbf0\}\times
	\prod_{j=i+1}^{t}\mathcal C_{I_j}
	\right)
	\subseteq \mathcal C_F.
\end{equation*}
For any fixed $\mathbf v\in\mathcal C_{I_i}$, the code in Lemma \ref{lem6} obtained by fixing the $i$th input is a coset of $\mathcal C_{i,\mathbf{0}}$. Translation preserves its ELP. Proposition \ref{prop1} therefore turns Lemmas \ref{lem6} and \ref{lem7} into
\begin{equation}\label{eqn20}
	|\mathcal U_i|
	\le
	\min\left\{
	n_{I_i},
	n_F-d_{k_{-i}}(\mathcal C_{i,\mathbf{0}}),
	n_F-d_{k_{-i}}(\mathcal C_F)
	\right\},
\end{equation}
which is \cite[Proposition~3]{gruica2026linear}.

Corollary \ref{cor1} reduces to \cite[Corollary~4]{gruica2026linear}, and Lemma \ref{lem9} reduces, after using $\sum_i k_i=k_F$, to
\begin{equation*}
	|\mathcal U|
	\le
	\frac{t n_F-k_F-t(d(\mathcal C_F)-1)}{t-1},
\end{equation*}
which is \cite[Proposition~7]{gruica2026linear}. The dual-distance conclusions in \cite[Corollaries~5 and~6]{gruica2026linear} also follow from \eqref{eqn20} and Wei duality. We do not state them as nonlinear results because an arbitrary nonlinear code has no canonical dual code or corresponding Wei duality.

The nonlinear statements contain information absent from these linear formulas. The levels $k_{-i}$ may be nonintegral, the sets $\mathcal C_{i,\mathbf v}$ need not be cosets or have the same ELP, and Lemma \ref{lem9} retains the separate terms $\lceil k_i\rceil$. In addition, Lemma \ref{lem8} gives joint restrictions for every selected collection of initial codes, rather than only an individual or total bound.

\subsection{Bounds on Read Symbols}

For a linear conversion, Proposition \ref{prop1} makes $\delta_\rho^{\elp}(\mathcal C_F)$ constant on each interval $\rho\in(j-1,j]$. Hence Lemma \ref{lem10} reduces to
\begin{equation}\label{eqn21}
	|\mathcal R_i|\ge\max\Bigl\{0,\max_{1\le j\le k_i}\bigl\{k_i-j+1-\bigl(|\mathcal U_i|-d_j(\mathcal C_F)+1\bigr)^+\bigr\}\Bigr\}.
\end{equation}
The inner maximum is the expression in \cite[Proposition~9]{gruica2026linear}. Let
\begin{equation*}
	j_i^*=\max\{j\in[k_i]:d_j(\mathcal C_F)\le|\mathcal U_i|\},
\end{equation*}
with $j_i^*=0$ if the set is empty. Equation \eqref{eqn17} gives $\theta_i=j_i^*$. The exact evaluation in Theorem \ref{thm1} then becomes
\begin{equation}\label{eqn22}
	|\mathcal R_i|\ge k_i-j_i^*,
\end{equation}
which is \cite[Theorem~10]{gruica2026linear}. In particular, the two ranges considered in the proof of Theorem \ref{thm1} correspond to the two index ranges in the proof of that theorem, with real entropy levels used here.

Finally, \cite[Corollary~11]{gruica2026linear} is the $j=1$ consequence of \eqref{eqn21}, using $d_1(\mathcal C_F)=d(\mathcal C_F)$. Combining it with the linear specialization of Corollary \ref{cor1} recovers \cite[Corollary~12]{gruica2026linear}. Thus all general conclusions about read symbols in \cite{gruica2026linear} are retained. For nonlinear codes, $\delta_\rho^{\elp}$ may change at nonintegral entropy levels, so neither the integer index $j$ nor $d_j$ alone describes the full optimization.

The Reed--Muller conversion in \cite[Theorem~14]{gruica2026linear} attains the unchanged-symbol bounds for both initial codes. In the parameter range analyzed in \cite[Section~IV-C]{gruica2026linear}, its read cost from the second initial code attains \cite[Theorem~10]{gruica2026linear} for $m\ge5$, whereas the first read cost satisfies only $|\mathcal R_1|\le k_1$ while the corresponding lower bound is $|\mathcal R_1|\ge m-1$. Thus that construction does not establish read optimality for every initial code. In contrast, Construction \ref{con1} below attains both \eqref{eqn20} and \eqref{eqn22} for every initial code.

\section{Constructions Attaining the Bounds}\label{sec6}

In this section, we first introduce a construction of convertible codes based on the well-known Hamming codes, which can attain the bounds given in both \eqref{eqn20} and \eqref{eqn22} for linear codes. We then propose a construction of convertible codes based on quasigroups, which can attain the bounds given in Lemma \ref{lem7} and Theorem \ref{thm2} for general nonlinear codes.

\subsection{Hamming Codes}\label{sec7}

We recall some properties of Hamming codes. Then, we establish the cardinality condition for converting copies of $\mathcal H_q(r)$ to $\mathcal H_q(m)$, and derive specific upper and lower bounds on unchanged and read symbols in Theorem \ref{thm2}.

Let $q$ be a prime power and $r\ge2$. The projective space $\operatorname{PG}(r-1,q)$ is the set of one-dimensional subspaces of $\mathbb F_q^r$. For every nonzero $\mathbf v\in\mathbb F_q^r$, denote the corresponding projective point by
\begin{equation*}
	[\mathbf v]=\{\lambda\mathbf v:\lambda\in\mathbb F_q^*\}.
\end{equation*}
Thus, $[\mathbf v]=[\mathbf w]$ if and only if $\mathbf w=\lambda\mathbf v$ for some $\lambda\in\mathbb F_q^*$. Choose one nonzero representative $\mathbf h_j$ from each projective point and form
\begin{equation}\label{eqn23}
	\mathbf H_r=\bigl(\mathbf h_1\mid\mathbf h_2\mid\cdots \mid\mathbf h_{N_r}\bigr),
\end{equation}
where $\{[\mathbf h_j]:j\in[N_r]\}=\operatorname{PG}(r-1,q)$.
The $[N_r,K_r,3]_q$ Hamming code is defined as
\begin{equation}\label{eqn24}
	\mathcal H_q(r)=\{\mathbf c\in\mathbb F_q^{N_r}:\mathbf H_r\mathbf c^{\top}=\mathbf0\}.
\end{equation}
Distinct columns of $\mathbf H_r$ are not proportional, so any two columns are linearly independent. On the other hand, there exist three linearly dependent columns, for example representatives of $[\mathbf e_1]$, $[\mathbf e_2]$, and $[\mathbf e_1+\mathbf e_2]$. Hence $d(\mathcal H_q(r))=3$. Its parameters are
\begin{equation}\label{eqn25}
	N_r=\frac{q^r-1}{q-1},\ K_r=N_r-r,
\end{equation}
and $|\operatorname{PG}(r-1,q)|=N_r$. For convenience, define $N_0=0$ and $N_1=1$. If $r$ linearly independent columns are placed last and transformed by row operations into the identity matrix $\mathbf I_r$, then an equivalent code has the systematic parity-check matrix
\begin{equation}\label{eqn26}
	\mathbf H_r=(\mathbf P\mid \mathbf I_r).
\end{equation}

\begin{theorem}\label{thm2}
	Let $2\le r<m$. An information-preserving merge conversion
	\begin{equation}\label{eqn27}
		\prod_{i=1}^{t}\mathcal H_q(r)\longrightarrow\mathcal H_q(m)
	\end{equation}
	exists if and only if
	\begin{equation}\label{eqn28}
		K_r\mid K_m,
		\quad t=\frac{K_m}{K_r},
	\end{equation}
where $\prod_{i=1}^{t}\mathcal H_q(r)$ is the Cartesian product of $t$ Hamming codes $\mathcal H_q(r)$. For every conversion in \eqref{eqn27},
	\begin{equation}\label{eqn29}
		|\mathcal U_i|\le K_r.
	\end{equation}
For every $i\in[t]$, the read cost satisfies
\begin{equation}\label{eqn30}
		|\mathcal R_i|
		\ge K_r-|\mathcal U_i|+\nu_m(|\mathcal U_i|),
\end{equation}
where $\nu_m(a) = \left|\left\{\ell\in\{0,\ldots,m-1\}: 1+N_\ell\le a\right\} \right|$. In particular, if $|\mathcal U_i|=K_r$, then
\begin{equation*}
    |\mathcal R_i| \ge
    \begin{cases}
        r, & (q,r)\ne(2,2),\\
        1, & (q,r)=(2,2).
    \end{cases}
\end{equation*}
\end{theorem}

\begin{IEEEproof}
	By \eqref{eqn25}, we have $|\prod_{i=1}^{t}\mathcal H_q(r)|=q^{tK_r}$ and $|\mathcal H_q(m)|=q^{K_m}$. Hence, $|\mathcal H_q(r)^t|=|\mathcal H_q(m)|$ if and only if
	\begin{equation*}
		tK_r=K_m,
	\end{equation*}
	which proves \eqref{eqn28} since $t$ is a positive integer.

	We next establish the bounds on unchanged and read symbols. The dual of $\mathcal H_q(m)$ is the $q$-ary Simplex code, and \cite[Example 7.10.13]{huffman2010fundamentals} shows that
	\begin{equation*}
		d_j(\mathcal H_q(m)^\perp)
		=\frac{q^m-q^{m-j}}{q-1},
		\quad j\in[m].
	\end{equation*}
	Then,
	\begin{equation*}
		N_m+1-d_j(\mathcal H_q(m)^\perp) = 1+N_{m-j},
		\quad j\in[m].
	\end{equation*}
	Wei duality \cite[Theorem 3]{wei1991generalized} gives
	\begin{equation}\label{eqn31}
		\{d_j(\mathcal H_q(m)):j\in[K_m]\}
		=[N_m]\setminus \{1+N_\ell:0\le\ell\le m-1\}.
	\end{equation}
	Let $x=N_m-K_r$. Since $N_m=N_{m-1}+q^{m-1}$ and $r<m$, we have
	\begin{equation*}
		x-(1+N_{m-1})
		=q^{m-1}-K_r-1
		>0.
	\end{equation*}
	Hence
	\begin{equation*}
		1+N_\ell < x,\quad 0\le \ell\le m-1
	\end{equation*}
	and $$x\in \{d_j(\mathcal H_q(m)):j\in[K_m]\}.$$ The number of GHWs not exceeding $x$ is
	\begin{align*}
		&\left| [x]\setminus \{1+N_\ell:0\le\ell\le m-1\} \right|\\
		&=x-m \\
		&=N_m-K_r-m \\
		&=K_m-K_r.
	\end{align*}
	Therefore,
	\begin{equation*}
		d_{K_m-K_r}(\mathcal H_q(m))
		=x
		=N_m-K_r.
	\end{equation*}
	Lemma \ref{lem7} and Proposition \ref{prop1} therefore yield
	\begin{equation*}
		|\mathcal U_i|
		\le N_m-d_{K_m-K_r}(\mathcal H_q(m))
		=K_r.
	\end{equation*}
	Moreover, \eqref{eqn31} also gives
	\begin{align*}
		&\max\{j:d_j(\mathcal H_q(m))\le |\mathcal U_i|\} \\
		&= |\mathcal U_i|-\left| \{\ell : 1+N_\ell\le|\mathcal U_i| :0\le\ell\le m-1\} \right| \\
		&=|\mathcal U_i|-\nu_m(|\mathcal U_i|).
	\end{align*}
	Proposition \ref{prop1} and Theorem \ref{thm1} now give
	\begin{equation*}
		|\mathcal R_i|
		\ge K_r-|\mathcal U_i|+\nu_m(|\mathcal U_i|).
	\end{equation*}
	 
	 Finally, setting $|\mathcal U_i|=K_r$ gives
	\begin{equation*}
    	|\mathcal R_i|\ge \nu_m(K_r).
	\end{equation*}
	By the definition of $\nu_m$,
	\begin{align*}
    	\nu_m(K_r)
    	&= \left| \left\{ \ell\in\{0,\ldots,m-1\}: 1+N_\ell\le K_r \right\} \right|\\
    	&= \begin{cases}
        		r, & (q,r)\ne(2,2),\\
        		1, & (q,r)=(2,2).
    		\end{cases}
	\end{align*}
	Indeed, when $(q,r)\ne(2,2)$,
	\begin{equation*}
    	1+N_{r-1}\le N_{r-1}+q^{r-1}-r = N_r-r = K_r<1+N_r,
	\end{equation*}
	so precisely $\ell=0,\ldots,r-1$ are counted, where equality holds for $(q,r)=(2,3)$ and $(q,r)=(3,2)$. When $(q,r)=(2,2)$, we have $K_2=N_2-2=1$, so only $\ell=0$ is counted.
	
	We complete the proof.
\end{IEEEproof}

To apply Theorem \ref{thm2}, the following lemma supplies the common matrix structure used in Subsection \ref{sec8}.

\begin{lemma}\label{lem11}
	Let $q$ be a prime power and $m\ge3$. There exist
	\begin{equation*}
		t=\frac{K_m}{q-1}
	\end{equation*}
	matrices $\mathbf P_i\in\mathbb F_q^{m\times(q-1)}$, $\mathbf B_i\in\mathbb F_q^{m\times2}$, and $\mathbf A_i\in\mathbb F_q^{2\times(q-1)}$, $i\in[t]$, such that
	\begin{equation}\label{eqn32}
		\mathbf P=(\mathbf P_1\mid\cdots\mid\mathbf P_t),
		\quad \mathbf P_i=\mathbf B_i\mathbf A_i.
	\end{equation}
	Moreover,
	\begin{equation}\label{eqn33}
		\mathbf H_F=(\mathbf P_1\mid\cdots\mid\mathbf P_t\mid \mathbf I_m)
	\end{equation}
	is a parity-check matrix of $\mathcal H_q(m)$, while
	\begin{equation}\label{eqn34}
		\mathbf H_{I_i}=(\mathbf A_i\mid \mathbf I_2)
	\end{equation}
	is a parity-check matrix of $\mathcal H_q(2)$ for every $i\in[t]$.
\end{lemma}

\begin{IEEEproof}
	Let $\mathbf e_1,\ldots,\mathbf e_m$ be the standard basis of $\mathbb F_q^m$. In \eqref{eqn23}, choose these vectors as the representatives of the last $m$ projective points. The remaining projective points form
	\begin{equation}\label{eqn35}
		\mathcal P=\operatorname{PG}(m-1,q)
		\setminus\{[\mathbf e_1],\ldots,[\mathbf e_m]\}.
	\end{equation}
	Every vector $\boldsymbol\gamma$ in $\mathbb{F}_q^m$ with $\supp(\boldsymbol\gamma)=\{j_1<\cdots<j_r\}\subseteq[m]$ has a unique expression
	\begin{equation*}
		\boldsymbol\gamma = \sum_{\ell=1}^{r} \alpha_{j_\ell} \mathbf e_{j_\ell}.
	\end{equation*}
	Therefore, every point $[\boldsymbol\gamma]$ in $\operatorname{PG}(m-1,q)$ has a unique expression
	\begin{equation*}
		[\boldsymbol\gamma] = \left[ \mathbf e_{j_1} + \sum_{\ell=2}^{r} \alpha_{j_1}^{-1}\alpha_{j_\ell} \mathbf e_{j_\ell} \right].
	\end{equation*}
	For convenience, we write $[\boldsymbol\gamma]$ in the form
	\begin{equation}\label{eqn36}
		[\boldsymbol\gamma] = \left[ \mathbf e_{j_1} + \sum_{\ell=2}^{r-1}\beta_\ell\mathbf e_{j_\ell} +\lambda \mathbf e_{j_r} \right],
	\end{equation}
	where $\beta_\ell = \alpha_{j_1}^{-1}\alpha_{j_\ell}$ for $\ell\in\{2,3,\dots,r-1\}$ and $\lambda=\alpha_{j_1}^{-1}\alpha_{j_r}$.

	For every
	\begin{equation*}
		S=\{j_1<\cdots<j_r\}\subseteq[m],\ r\ge2,
	\end{equation*}
	and
	\begin{equation*}
		\boldsymbol\beta=(\beta_2,\ldots,\beta_{r-1})\in(\mathbb F_q^*)^{r-2},
	\end{equation*}
	define
	\begin{equation}\label{eqn37}
		\mathbf v(S^-,\boldsymbol\beta)
		=\mathbf e_{j_1}+
		\sum_{\ell=2}^{r-1}\beta_\ell\mathbf e_{j_\ell}
	\end{equation}
	and
	\begin{equation}\label{eqn38}
		\mathcal P(S,\boldsymbol\beta)
		=\{[\mathbf v(S^-,\boldsymbol\beta)+\lambda\mathbf e_{j_r}]:
		\lambda\in\mathbb F_q^*\},
	\end{equation}
	where $S^-=S\setminus\{j_r\}$. Note that
	\begin{equation*}
		[\mathbf v(S^-,\boldsymbol\beta)+\lambda\mathbf e_{j_r}]\ne [\mathbf v(S'^-,\boldsymbol\beta')+\lambda' \mathbf e_{j_r'}] \
		\begin{cases}
			\mbox{if } S \ne S', \\
			\mbox{if } S = S', \boldsymbol\beta\ne \boldsymbol\beta', \\
			\mbox{if } S=S', \boldsymbol\beta= \boldsymbol\beta', \lambda\ne \lambda'.
		\end{cases}
	\end{equation*}
	Hence
	\begin{equation}\label{eqn39}
		\mathcal P
		=\bigsqcup_{\substack{S\subseteq[m]\\|S|\ge2}}
		\bigsqcup_{\boldsymbol\beta\in(\mathbb F_q^*)^{|S|-2}}
		\mathcal P(S,\boldsymbol\beta),
	\end{equation}
	and
	\begin{equation*}
		|\mathcal P(S,\boldsymbol\beta)|=q-1.
	\end{equation*}
	The number of disjoint sets in \eqref{eqn39} is
	\begin{align}
		\sum_{r=2}^m\binom mr(q-1)^{r-2}
		&=\frac{q^m-1-m(q-1)}{(q-1)^2}\nonumber\\
		&=\frac{K_m}{q-1}=t.\label{eqn40}
	\end{align}
	In particular, this calculation also proves that $t$ is an integer. 

	Label the disjoint sets in \eqref{eqn39} as $\mathcal P_1,\ldots,\mathcal P_t$, and suppose $\mathbb F_q^*=\{\lambda_1,\ldots,\lambda_{q-1}$\}. Define
	\begin{equation}\label{eqn41}
		\mathbf B_i=(\mathbf v(S^-,\boldsymbol\beta)\mid\mathbf e_{j_r}) \in \mathbb{F}_q^{m\times 2}.
	\end{equation}
	and
	\begin{equation}\label{eqn42}
		\mathbf A_i=
		\begin{pmatrix}
		1&1&\cdots&1\\
		\lambda_1&\lambda_2&\cdots&\lambda_{q-1}
		\end{pmatrix} \in \mathbb{F}_q^{2\times (q-1)}.
	\end{equation}
	Let $\mathbf P_i=\mathbf B_i \mathbf A_i$, where the $j$th column is
	\begin{equation*}
		\mathbf B_i \cdot (1,\lambda_j)^{\top}
		=\mathbf v(S^-,\boldsymbol\beta)+\lambda_j\mathbf e_{j_r}.
	\end{equation*}
	
	By \eqref{eqn39}, the columns of $\mathbf P_1,\ldots,\mathbf P_t$ represent all points in $\mathcal P$ exactly once. The columns of $\mathbf I_m$ represent the remaining points $[\mathbf e_1],\ldots,[\mathbf e_m]$. Hence the columns of $\mathbf H_F$ in \eqref{eqn33} represent all points of $\operatorname{PG}(m-1,q)$ exactly once, and \eqref{eqn23} shows that $\mathbf H_F$ is a parity-check matrix of $\mathcal H_q(m)$. Note that
	\begin{equation}\label{eqn43}
		\mathbf H_{I_i}=
		\left(
		\begin{array}{cccc|cc}
		1&1&\cdots&1&1&0\\
		\lambda_1&\lambda_2&\cdots&\lambda_{q-1}&0&1
		\end{array}
		\right).
	\end{equation}
	Its columns represent $[(1,\lambda)^{\top}]$, $\lambda\in\mathbb F_q^*$, together with $[(1,0)^{\top}]$ and $[(0,1)^{\top}]$. These are all points of $\operatorname{PG}(1,q)$, so $\mathbf H_{I_i}$ is a parity-check matrix of $\mathcal H_q(2)$.
\end{IEEEproof}

\begin{remark}
	In fact, the proof of Lemma \ref{lem11} admits a geometric interpretation of the partition, illustrated in Figure \ref{fig1}. Each projective line corresponds to an initial code. Each set $\mathcal P(S,\boldsymbol\beta)$ consists of the $q-1$ points on the line determined by $[\mathbf v(S^{-},\boldsymbol\beta)]$ and $[\mathbf e_{j_r}]$, excluding these two endpoints. The complete lines need not be disjoint. The retained sets form a partition because the normalized representative \eqref{eqn36} of every nonbasis projective point uniquely determines $S$, $\boldsymbol\beta$.
\end{remark}

\begin{figure}[t]
\centering
\begin{tikzpicture}
  \draw[thick] (-1.29,3.2) -- (0.94,0);
  \draw[thick] (-1.6,2.11) -- (1.82,4.8);
  \draw[red, fill=red] (-0.9,2.64) circle (0.15cm);
  \draw[thick] (0.2,0.27) -- (5.05,0.45);
  \draw[red, fill=red] (0.75,0.32) circle (0.15cm);
  \draw[thick] (0.82,5) -- (4.91,0.17);
  \draw[red, fill=red] (1.3,4.45) circle (0.15cm);
  \draw[red, fill=red] (4.7,0.47) circle (0.15cm);

  \node[draw=none, text=red, minimum width=20pt] at (1.38,5.16) {$[\mathbf{e}_{j_{r}}]$};
  \node[draw=none, text=red, minimum width=20pt] at (0.45,-0.18) {$[\mathbf{e}_{j_{r}'}]$};
  \node[draw=none, text=red] at (-2.23,2.73) {$[\mathbf{v}( S^{-} ,\boldsymbol{\beta}')]$};
  \node[draw=none, text=red] at (6.02,0.12) {$[\mathbf{v}( S^{-} ,\boldsymbol{\beta})]$};

  \draw[fill=blue, blue] (-0.21,3.2) circle (0.14cm);
  \draw[fill=blue, blue] (0.28,3.6) circle (0.16cm);
  \draw[fill=blue, blue] (0.76,4) circle (0.15cm);
  \node[draw=none, text=blue] at (-0.21,4.25) {$\mathcal{P}_1$};

  \draw[fill=blue, blue] (1.98,3.65) circle (0.15cm);
  \draw[fill=blue, blue] (2.9,2.57) circle (0.15cm);
  \draw[fill=blue, blue] (3.71,1.55) circle (0.15cm);
  \node[draw=none, text=blue] at (3.41,3.06) {$\mathcal{P}_2$};

  \draw[fill=blue, blue] (1.68,0.33) circle (0.15cm);
  \draw[fill=blue, blue] (3.59,0.43) circle (0.15cm);
  \draw[fill=blue, blue] (2.5,0.35) circle (0.15cm);
  \node[draw=none, text=blue] at (2.9,-0.19) {$\mathcal{P}_3$};

  \draw[fill=blue, blue] (-0.54,2.12) circle (0.15cm);
  \draw[fill=blue, blue] (-0.1,1.45) circle (0.15cm);
  \draw[fill=blue, blue] (0.3,0.86) circle (0.15cm);
  \node[draw=none, text=blue] at (-0.63,1.15) {$\mathcal{P}_4$};
\end{tikzpicture}
\caption{The red points are the endpoints excluded from the corresponding sets, and the blue points form $\mathcal P_1,\ldots,\mathcal P_4$. }
\label{fig1}
\end{figure}
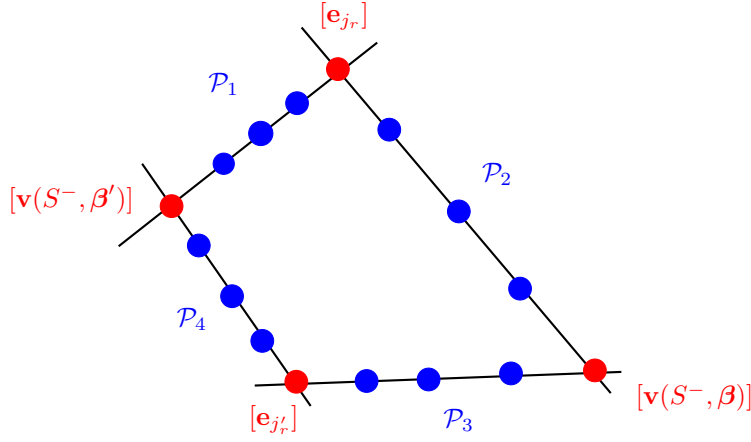

\subsection{A Linear Hamming Code Conversion}\label{sec8}

Condition \eqref{eqn28} guarantees the existence of an information-preserving bijection, but it does not ensure that the conversion is optimal. The following construction gives an explicit optimal conversion for $r=2$.

\begin{construction}\label{con1}
	Let $q\ge4$, $m\ge3$, and let $\mathbf A_i$, $\mathbf B_i$, and $\mathbf P_i$, $i\in[t]$, be the matrices in Lemma \ref{lem11}.

	\textit{Initial codes:} For every $i\in[t]$, let
	\begin{equation*}
		\mathcal C_{I_i}=\ker \mathbf H_{I_i}=\mathcal H_q(2).
	\end{equation*}

	\textit{Conversion:} Write the $i$th initial codeword as $(\mathbf x_i,\mathbf p_i)$, where $\mathbf x_i\in\mathbb F_q^{q-1}$ and $\mathbf p_i\in\mathbb F_q^2$. Keep all $q-1$ symbols of $\mathbf x_i$, read the two symbols of $\mathbf p_i$, and compute
	\begin{equation}\label{eqn44}
		\mathbf p_F^{\top}=\sum_{i=1}^t\mathbf B_i\mathbf p_i^{\top}.
	\end{equation}

	\textit{Final code:} The final code is defined as
	\begin{equation}\label{eqn45}
		\mathcal C_F = \{ (\mathbf x_1,\ldots,\mathbf x_t,\mathbf p_F): (\mathbf x_i,\mathbf p_i)\in \mathcal C_{I_i}, i\in[t] \}.
	\end{equation}
\end{construction}

\begin{corollary}\label{cor3}
	Construction \ref{con1} is an optimal nontrivial convertible code from $\mathcal H_q(2)$ to $\mathcal H_q(m)$.
\end{corollary}

\begin{IEEEproof}
	Since $\mathbf H_{I_i}$ is a parity-check matrix of $\mathcal C_{I_i}=\mathcal H_q(2)$ by Lemma \ref{lem11}, we have
	\begin{equation}\label{eqn46}
		\mathbf H_{I_i}\cdot (\mathbf x_i,\mathbf p_i)^{\top}=\mathbf0,
		\quad
		\mathbf p_i^{\top}=-\mathbf A_i\mathbf x_i^{\top}.
	\end{equation}
	By \eqref{eqn46}, \eqref{eqn44}, and $\mathbf P_i=\mathbf B_i\mathbf A_i$, we obtain
	\begin{equation*}
		\mathbf p_F^{\top}
		=-\sum_{i=1}^t\mathbf B_i\mathbf A_i\mathbf x_i^{\top}
		=-\sum_{i=1}^t\mathbf P_i\mathbf x_i^{\top},
	\end{equation*}
	which implies that every final codeword satisfies
	\begin{equation*}
		\mathbf H_F\cdot (\mathbf x_1,\ldots,\mathbf x_t, \mathbf p_F)^{\top}=\mathbf0.
	\end{equation*}
	Thus, $\mathcal C_F\subseteq\ker\mathbf H_F$. The conversion map is bijective and $\dim(\ker\mathbf H_F) = K_m= t(q-1)$. Therefore,
	\begin{equation*}
		|\mathcal C_F|=q^{t(q-1)}=|\ker\mathbf H_F|.
	\end{equation*}
	Hence, $\mathcal C_F\supseteq\ker\mathbf H_F$. Together with Lemma \ref{lem11}, we have
	\begin{equation*}
		\mathcal C_F=\ker \mathbf H_F=\mathcal H_q(m).
	\end{equation*}
	Theorem \ref{thm2} with $r=2$ gives $|\mathcal U_i|\le q-1$ and $|\mathcal R_i|\ge2$. Both equations $|\mathcal U_i|= q-1$ and $|\mathcal R_i|=2$ are attained in \eqref{eqn44}. Moreover, $q\ge4$ implies $2<K_2=q-1$, so the conversion is nontrivial.
\end{IEEEproof}

\begin{remark}
Corollary \ref{cor3} shows a simultaneously optimal conversion outside the MDS setting. This differs from the Reed--Muller construction in \cite[Section~IV]{gruica2026linear}, for which the bound on $|\mathcal R_1|$ is not shown to be attained.
\end{remark}

Construction \ref{con1} uses matrices $\mathbf P_1,\mathbf P_2,\dots,\mathbf P_t$ in Lemma \ref{lem11} exactly once. We next consider conversions obtained by selecting only part of these matrices or by using some matrices repeatedly. The following construction uses the same conversion procedure \eqref{eqn44} as Construction \ref{con1} with $\mathbf B_i$ replaced by $\mathbf B_{\varphi(i)}$.

\begin{construction}\label{conB}
	Under the assumptions of Construction \ref{con1}, let $t_0=K_m/(q-1)$ be the number of matrices in Lemma \ref{lem11}. Let $\varphi:[t]\to[t_0]$ be injective when $1\le t<t_0$ and surjective when $t>t_0$.
	
	\textit{Initial codes:} For every $i\in[t]$, let
	\begin{equation*}
		\mathcal C_{I_i}=\ker \mathbf H_{I_i}=\mathcal H_q(2),\
		\mathbf H_{I_i}=(\mathbf A_{\varphi(i)}\mid\mathbf I_2).
	\end{equation*}

	\textit{Conversion:} Write the $i$th initial codeword as $(\mathbf x_i,\mathbf p_i)$, where $\mathbf x_i\in\mathbb F_q^{q-1}$ and $\mathbf p_i\in\mathbb F_q^2$. Keep all $q-1$ symbols of $\mathbf x_i$, read the two symbols of $\mathbf p_i$, and compute
	\begin{equation*}
		\mathbf p_F^{\top}=\sum_{i=1}^t\mathbf B_{\varphi(i)}\mathbf p_i^{\top}.
	\end{equation*}

	\textit{Final code:} The final code is defined as
	\begin{equation*}
		\mathcal C_F = \{ (\mathbf x_1,\ldots,\mathbf x_t,\mathbf p_F): (\mathbf x_i,\mathbf p_i)\in \mathcal C_{I_i}, i\in[t] \}.
	\end{equation*}
\end{construction}

The following corollary gives sufficient conditions under which Construction \ref{conB} continues to attain bounds for unchanged and read symbols.

\begin{corollary}\label{cor4}
	Define $\mathbf H_F^{(t)}=(\mathbf P_{\varphi(1)}\mid\cdots\mid\mathbf P_{\varphi(t)}\mid\mathbf I_m)$.
	Then, the final code $\mathcal C_F$ given by Construction \ref{conB} is $\mathcal C_F=\ker\mathbf H_F^{(t)}$.
	For each $i\in[t]$, suppose that there is no zero row in the matrix
	\begin{equation}\label{eqn47}
		\bigl(\mathbf P_{\varphi(1)}\mid \cdots \mid \mathbf P_{\varphi(i-1)}\mid \mathbf P_{\varphi(i+1)}\mid \cdots \mid \mathbf P_{\varphi(t)}\bigr).
	\end{equation}
	When $t>t_0$, suppose further that
	\begin{equation}\label{eqn48}
		\max_{s\in[t_0]}|\varphi^{-1}(s)|\le q-2.
	\end{equation}
	Construction \ref{conB} simultaneously attains the bounds on the number of unchanged symbols in Lemma \ref{lem6} and of read symbols in Theorem \ref{thm1}. 
\end{corollary}

\begin{IEEEproof}	
	The identity $\mathcal C_F=\ker\mathbf H_F^{(t)}$ follows directly from the proof of Corollary \ref{cor3}. It remains to show that the bounds on unchanged and read symbols are both attained.
		
	Fix $i\in[t]$ and $\mathbf v\in\mathcal C_{I_i}$. When the $i$th initial codeword is fixed to $\mathbf v$, varying the other initial codewords gives
	\begin{equation}\label{eqn49}
		\mathcal C_{i,\mathbf v}=\mathbf c_{\mathbf v}+\mathcal D_i,
	\end{equation}
	where $\mathbf c_{\mathbf v}$ is fixed and
	\begin{equation*}
		\mathcal D_i
		=\left \{\left( 
		(\mathbf x_1,\dots,\mathbf x_{i-1}, \mathbf 0, \mathbf x_{i+1},\dots, \mathbf x_t), -\sum_{j\in[t]\setminus\{i\}}\mathbf x_j\mathbf P_{\varphi(j)}^{\top} 
		\right) : \mathbf x_j\in\mathbb F_q^{q-1} \right\}.
	\end{equation*}
	Every systematic coordinate associated with $\mathbf x_j$, $j\ne i$, belongs to $\supp(\mathcal D_i)$. Moreover, the $\ell$-th parity coordinate belongs to $\supp(\mathcal D_i)$ if and only if
	\begin{equation*}
		\bigl(\mathbf P_{\varphi(j)}\bigr)_{\ell,a}\ne0
	\end{equation*}
	for some $j\in[t]\setminus\{i\}$ and $a\in[q-1]$, where $\bigl(\mathbf P_{\varphi(j)}\bigr)_{\ell,a}$ denotes the entry in the $\ell$-th row and the $a$-th column of the matrix $\mathbf P_{\varphi(j)}$. The assumed condition \eqref{eqn47} implies
	\begin{equation*}
		|\supp(\mathcal D_i)|=n_F-(q-1).
	\end{equation*}
	Since $\dim(\mathcal D_i)=k_{-i}$, we have
	\begin{equation*}
		\delta_{k_{-i}}^{\elp}(\mathcal C_{i,\mathbf v})
		= \delta_{k_{-i}}^{\elp}(\mathcal D_i)
		= d_{k_{-i}}(\mathcal D_i)
		= |\supp(\mathcal D_i)|
		= n_F-(q-1),
	\end{equation*}
	where the first equality follows from  \eqref{eqn49}, and the second follows from Proposition \ref{prop1}. Therefore, Lemma \ref{lem6} shows
	\begin{equation*}
		|\mathcal U_i|\le q-1.
	\end{equation*}
	Since the conversion keeps exactly $q-1$ symbols from every initial codeword, it attains the bound on unchanged symbols.

	We next prove that the read bound is also attained. Let $L_i$ be the set of the $q-1$ coordinates corresponding to the columns of $\mathbf P_{\varphi(i)}$. The subcode of $\mathcal C_F$ supported on $L_i$ satisfies
	\begin{equation*}
		\mathcal C_F(L_i)|_{L_i}
		=\ker\mathbf P_{\varphi(i)}.
	\end{equation*}
	Since $\mathbf B_{\varphi(i)}$ has full column rank, we have $\operatorname{rank}(\mathbf P_{\varphi(i)})=\operatorname{rank}(\mathbf B_{\varphi(i)}\mathbf A_{\varphi(i)})=\operatorname{rank}(\mathbf A_{\varphi(i)})=2$. Thus, the rank--nullity theorem gives
	\begin{equation*}
		\dim\mathcal C_F(L_i)=(q-1)-2=q-3.
	\end{equation*}
	Consequently,
	\begin{equation}\label{eqnGHW1}
	\begin{split}
		d_{q-3}(\mathcal C_F)
		&\le|\supp(\mathcal C_F(L_i))|\\
		&\le |L_i| \\
		&=q-1.
	\end{split}
	\end{equation}

	We now show that $d_{q-2}(\mathcal C_F)>q-1$. It suffices to show that there is no $(q-2)$-dimensional subcode of $\mathcal C_F$ whose support has size at most $q-1$, which is equivalent to showing that $\dim\mathcal C_F(J)<q-2$ for every coordinate set $J$ with $|J|\le q-1$. If $|J|<q-2$, then
	\begin{equation*}
    	\dim\mathcal C_F(J)\le |J|<q-2
	\end{equation*}
	holds immediately. Hence it remains only to consider
	\begin{equation*}
    	q-2\le |J|\le q-1.
	\end{equation*}
	If $|J|=q-2$, every column of $\mathbf H_F^{(t)}$ is nonzero, and hence
	\begin{align*}
		\dim\mathcal C_F(J)
		&=|J|-\operatorname{rank}\bigl(\mathbf H_F^{(t)}|_J\bigr)\\
		&\le(q-2)-1 \\
		&=q-3,
	\end{align*}
	where $\mathbf H_F^{(t)}|_J$ denotes the submatrix of $\mathbf H_F^{(t)}$ formed by the columns indexed by $J$. If $|J|=q-1$, then
	\begin{equation*}
		\operatorname{rank}
		\bigl(\mathbf H_F^{(t)}|_J\bigr)\ge2.
	\end{equation*}
	Indeed, otherwise all $q-1$ columns indexed by $J$ would represent the same projective point. This is impossible when $t<t_0$, since every projective point occurs at most once. It is also impossible when $t>t_0$, since \eqref{eqn48} implies that every projective point has multiplicity at most $q-2$. Therefore,
	\begin{align*}
		\dim\mathcal C_F(J)
		&=|J|-\operatorname{rank}\bigl(\mathbf H_F^{(t)}|_J\bigr)\\
		&\le q-3.
	\end{align*}
	Thus no subcode of dimension $q-2$ can have support of size at most $q-1$, which proves
	\begin{equation}\label{eqnGHW2}
		d_{q-2}(\mathcal C_F)>q-1.
	\end{equation}

	Combining the two GHW inequalities \eqref{eqnGHW1} and \eqref{eqnGHW2} gives
	\begin{equation*}
		\theta_i
		=\max\{j:d_j(\mathcal C_F)\le|\mathcal U_i|=q-1\}
		=q-3.
	\end{equation*}
	Therefore, Theorem \ref{thm1} implies
	\begin{equation*}
		|\mathcal R_i|
		\ge K_2-\theta_i
		=(q-1)-(q-3)
		=2.
	\end{equation*}
	The conversion reads exactly two symbols from every initial codeword and hence attains the read bound.
\end{IEEEproof}

\subsection{A Nonlinear Quasigroup Conversion}\label{sec9}

We next give a separate construction to show that the nonlinear bounds can be attained at nonintegral entropy levels. A binary operation $*:\mathcal A^2\to\mathcal A$ is a quasigroup operation if, for every $a,b\in\mathcal A$, each of the equations $a*x=b$ and $y*a=b$ has a unique solution. Equivalently, every row and every column of its operation table is a permutation of $\mathcal A$ \cite[Section~1]{mckay2007small}. More generally, $f:\mathcal A^r\to\mathcal A$ is an $r$-ary quasigroup if any one of the variables in $x_0=f(x_1,\ldots,x_r)$ is uniquely determined by the other $r$ variables \cite[Section~1.2]{krotov2016transitive}.

Let $\mathcal A=\{0,1,2,3,4\}$ and define $*$ by
\begin{equation}\label{eqn50}
\begin{array}{c|ccccc}
 *&0&1&2&3&4\\ \hline
 0&0&1&2&3&4\\
 1&1&0&3&4&2\\
 2&2&3&4&0&1\\
 3&3&4&1&2&0\\
 4&4&2&0&1&3
\end{array}
\end{equation}
Every row and every column contains every element of $\mathcal A$ exactly once. Hence $*$ is a quasigroup operation. For $r\ge2$, define
\begin{align}
	f_2(x_1,x_2)&=x_1*x_2,\nonumber\\
	f_r(x_1,\ldots,x_r)&=f_{r-1}(x_1,\ldots,x_{r-1})*x_r.
	\label{eqn51}
\end{align}
The composition result in \cite[Proposition~2-2(a)]{krotov2008decomposability} shows that $f_r$ is an $r$-ary quasigroup. This fact also follows directly by induction from the unique solvability of each equation involving $*$.

Two codes $\mathcal C,\mathcal D\subseteq\mathcal A^n$ are said to be equivalent if there exist a coordinate permutation $\pi$ of $[n]$ and alphabet permutations $\tau_1,\ldots,\tau_n$ of $\mathcal A$ such that
\begin{equation}\label{eqn52}
	\mathcal D=
	\left\{
	\bigl(\tau_1(c_{\pi(1)}),\ldots,\tau_n(c_{\pi(n)})\bigr):
	\mathbf c\in\mathcal C
	\right\}.
\end{equation}
Such an equivalence preserves cardinality and Hamming distance.

\begin{construction}\label{con2}
	Let $m\ge3$, $t\ge2$, and $\mathcal S=\{0,1,2\}\subset\mathcal A$.

	\textit{Initial codes:} For every $i\in[t]$, define
	\begin{equation}\label{eqn53}
		\mathcal C_{I_i} = \bigl\{ (\mathbf x_i, p_i): \mathbf x_i\in\mathcal S^m, p_i=f_m(\mathbf x_i) \bigr\}.
	\end{equation}

	\textit{Conversion:} Keep the $m$ coordinates in $\mathbf x_i$, read $p_i$ from every initial word, and use \eqref{eqn51} to compute
	\begin{equation}\label{eqn54}
		p_F=f_t(p_1,\ldots,p_t).
	\end{equation}
	
	\textit{Final code:} The final code is defined as
	\begin{equation}\label{eqn55}
		\mathcal C_F=\bigl\{(\mathbf x_1,\ldots,\mathbf x_t,p_F): \mathbf x_i\in\mathcal S^m, i\in[t] \bigr\}.
	\end{equation}
\end{construction}

\begin{proposition}\label{prop2}
	Construction \ref{con2} is a nontrivial convertible code from $(m+1,3^m,2)_5$ codes to a $(tm+1,3^{tm},2)_5$ code. Neither an initial code nor the final code is equivalent to a linear code over $\mathbb F_5$. More strongly, neither code is contained in a code equivalent to a linear code $[m+1,m,2]_5$ for each initial code and $[mt+1,mt,2]_5$ for the final code.
\end{proposition}

\begin{IEEEproof}
	The recursive composition in \eqref{eqn51} shows that the final parity is an $mt$-ary quasigroup function of the systematic coordinates, which implies that each initial code and the final code have minimum Hamming distance $2$. Therefore, each initial code has parameters $(m+1,3^m,2)_5$ and the final code has parameters $(tm+1,3^{tm},2)_5$ since the conversion map is bijective.

	Furthermore,
	\begin{equation*}
		\sum_{i=1}^t|\mathcal R_i|=t
		<t m\log_5 3=k_F,
	\end{equation*}
	so the conversion is nontrivial. Since neither $3^m$ nor $3^{mt}$ is an integral power of five, the code cardinalities preclude equivalence with  linear codes over $\mathbb F_5$. The stronger noncontainment statement is proved in Appendix \ref{app6}.
\end{IEEEproof}

In what follows, we will prove Construction \ref{con2} is optimal. To obtain explicit expressions for $\theta_i$ and $\delta_{k_{-i}}^{\elp}(\mathcal C_F)$, we first analyze the ELP of the final code.

\begin{proposition}\label{prop3}
	Let $\mathcal C_F$ be the final code in Construction \ref{con2}. Then, $h_m(\mathcal C_F)=(m-1)a$ and $\delta_{k_{-i}}^{\elp}(\mathcal C_F)=(t-1)m+1$, where $a=\log_5 3$.
\end{proposition}

\begin{IEEEproof}
The uniform distribution on $\mathcal C_F$ induces the uniform distribution on $\mathcal S^{mt}$ over the systematic coordinates. Hence
\begin{equation}\label{eqn56}
	\Pr\!\left(X_{[mt]}^{(\mathcal C_F)}=\mathbf x\right)=3^{-mt},
	\ \mathbf x\in\mathcal S^{mt},
\end{equation}
and the systematic coordinates are mutually independent, each with entropy $a$ since $\Pr\!\left(X_{\{i\}}^{(\mathcal C_F)}=x\right)=\frac{1}{3}$ for $i\in[mt]$ and $x\in\mathcal S$.

First, we determine the maximum residual entropy supported on any set of $m$ coordinates. Let $J\subseteq[mt+1]$ with $|J|=m$, and fix $\mathbf y\in\mathcal C_F|_{J^c}$. If $mt+1\in J$, the parity symbol is determined by the remaining systematic symbols in $J$. If $mt+1\notin J$, the parity symbol is fixed by $\mathbf y$, and the quasigroup property implies that any $m-1$ systematic symbols in $J$ uniquely determine the remaining one. Therefore, in either case, at most $m-1$ coordinates in $J$ can be chosen freely. Moreover,
\begin{equation*}
	\left| \left\{ \mathbf c|_J: \mathbf c\in\mathcal C_F,\  \mathbf c|_{J^c}=\mathbf y \right\} \right| \le |\mathcal S|^{m-1}.
\end{equation*}
Consequently,
\begin{align*}
	H\left(	X_J^{(\mathcal C_F)} \middle| X_{J^c}^{(\mathcal C_F)}=\mathbf y \right)
 	\le \log_q|\mathcal S|^{m-1}
 	=(m-1)a.
\end{align*}
Averaging over $\mathbf y\in\mathcal C_F|_{J^c}$ gives
\begin{equation*}
	H\left( X_J^{(\mathcal C_F)} \middle| X_{J^c}^{(\mathcal C_F)} \right) \le (m-1)a.
\end{equation*}
To attain equality, take $J_0=\{1,\ldots,m-1,mt+1\}$. Since the first $m-1$ systematic coordinates are independent of the remaining systematic coordinates and the parity symbol is determined by all systematic coordinates,
\begin{align*}
	H\left(X_{J_0}^{(\mathcal C_F)}\middle|X_{J_0^c}^{(\mathcal C_F)}\right)
	&=H\left(X_{[m-1]}^{(\mathcal C_F)},X_{\{mt+1\}}^{(\mathcal C_F)}\middle|X_{J_0^c}^{(\mathcal C_F)}\right) \\
	&=H\left(X_{[m-1]}^{(\mathcal C_F)}\middle|X_{J_0^c}^{(\mathcal C_F)}\right) + H\left(X_{\{mt+1\}}^{(\mathcal C_F)}\middle|X_{[m-1]}^{(\mathcal C_F)},X_{J_0^c}^{(\mathcal C_F)}\right) \\
	&=H\left(X_{[m-1]}^{(\mathcal C_F)}\middle|X_{J_0^c}^{(\mathcal C_F)}\right) \\
	&=(m-1)a.
\end{align*}
It follows that
\begin{equation}\label{eqn57}
	h_m(\mathcal C_F)
	=\underset{J\subseteq[mt+1], |J|=m}{\max} H\left( X_J^{(\mathcal C_F)} \middle| X_{J^c}^{(\mathcal C_F)} \right)
	=(m-1)a.
\end{equation}

Finally, we show that $\delta_{k_{-i}}^{\elp}(\mathcal C_F)=(t-1)m+1$. Note that $k_{-i}=(t-1)ma$ for every $i\in[t]$. Let $J_i$ contain the $m(t-1)$ systematic coordinates outside block $i$ and the parity coordinate. Thus, \eqref{eqn56} and the chain rule yield
\begin{align*}
	H\left(X_{J_i}^{(\mathcal C_F)}\middle|X_{J_i^c}^{(\mathcal C_F)}\right)
	&=H\left(X_{J_i\setminus\{mt+1\}}^{(\mathcal C_F)},X_{\{mt+1\}}^{(\mathcal C_F)} \middle|X_{J_i^c}^{(\mathcal C_F)}\right) \\
	&=H\left(X_{J_i\setminus\{mt+1\}}^{(\mathcal C_F)} \middle|X_{J_i^c}^{(\mathcal C_F)}\right) + H\left(X_{\{mt+1\}}^{(\mathcal C_F)} \middle|X_{J_i\setminus\{mt+1\}}^{(\mathcal C_F)},X_{J_i^c}^{(\mathcal C_F)}\right) \\
	&=H\left(X_{J_i\setminus\{mt+1\}}^{(\mathcal C_F)} \middle|X_{J_i^c}^{(\mathcal C_F)}\right) \\
	&=(t-1)ma.
\end{align*}
Therefore,
\begin{equation}\label{eqn58}
	h_{(t-1)m+1}(\mathcal C_F) = h_{|J_i|}(\mathcal C_F) \ge k_{-i}.
\end{equation}
Conversely, let $J\subseteq[mt+1]$ have size $(t-1)m$. If $mt+1\notin J^c$, then $J^c$ contains $m+1$ systematic coordinates and
\begin{equation*}
	H\left(X_{J^c}^{(\mathcal C_F)}\right)=(m+1)a>ma.
\end{equation*}
If $mt+1\in J^c$, choose a systematic coordinate $\ell\in J$. Then
\begin{align*}
	H\left(X_{J^c}^{(\mathcal C_F)}\right)
	&=H\left(X_{J^c\setminus \{mt+1\}}^{(\mathcal C_F)}\right)+ H\left(X_{\{mt+1\}}^{(\mathcal C_F)} \middle|X_{J^c\setminus\{mt+1\}}^{(\mathcal C_F)}\right)\\
	&=ma+H\left(X_{\{mt+1\}}^{(\mathcal C_F)} \middle|X_{J^c\setminus\{mt+1\}}^{(\mathcal C_F)}\right)\\
	&\ge ma+H\left(X_{\{mt+1\}}^{(\mathcal C_F)} \middle|X_{[mt]\setminus\{\ell\}}^{(\mathcal C_F)}\right)\\
	&=ma+a \\
	&>ma.
\end{align*}
Thus,
\begin{align*}
	H\left(X_J^{(\mathcal C_F)}\middle|X_{J^c}^{(\mathcal C_F)}\right)
	&=H\left(X^{(\mathcal C_F)}\right)
	-H\left(X_{J^c}^{(\mathcal C_F)}\right)\\
	&<(t-1)ma \\
	&=k_{-i}.
\end{align*}
Thus, $h_{(t-1)m}(\mathcal C_F)=h_{|J|}(\mathcal C_F)<k_{-i}$. Together with \eqref{eqn58}, we have
\begin{equation*}
	\delta_{k_{-i}}^{\elp}(\mathcal C_F)=(t-1)m+1.
\end{equation*}
We complete the proof.
\end{IEEEproof}

\begin{corollary}\label{cor5}
	Construction \ref{con2} is optimal, with $|\mathcal U_i|=m$ and $|\mathcal R_i|=1$ for every $i\in[t]$.
\end{corollary}

\begin{IEEEproof}
	Lemma \ref{lem7} and Proposition \ref{prop3} give
	\begin{equation*}
		|\mathcal U_i|\le(mt+1)-((t-1)m+1)=m.
	\end{equation*}
	The construction attains equality. According to \eqref{eqn11} and \eqref{eqn57}, we have
	\begin{equation*}
		\theta_i=\min\{ma,h_m(\mathcal C_F)\}=(m-1)a.
	\end{equation*}
	Theorem \ref{thm1} yields
	\begin{equation*}
		|\mathcal R_i|\ge
		\left\lceil ma-(m-1)a\right\rceil
		=\lceil a\rceil=1.
	\end{equation*}
	Thus the construction also attains the read bound.
\end{IEEEproof}

The dimensions $m\log_5 3$ and $mt\log_5 3$, as well as the threshold $\theta_i=(m-1)\log_5 3$, are nonintegral. Hence Corollary \ref{cor5} gives an explicit equality case that cannot be expressed by an integer-indexed GHW alone. The construction is included only to establish the tightness of the nonlinear bounds. Its minimum distance is two, and obtaining analogous equality cases with larger distance remains open.

\section{Conclusion}\label{sec10}

This paper established an information-theoretic framework for code conversion in the merge regime without assuming linearity. The inverse ELP $\delta_\rho^{\elp}$ allows the relevant entropy level $\rho$ to be real and therefore applies when a code has nonintegral dimension $\log_qM$. Conditional entropy yields individual and joint upper bounds on unchanged symbols and lower bounds on read cost. The minimum distance consequences further show that reading fewer than $\lceil k_i\rceil$ symbols from an initial code requires at least $d(\mathcal C_F)$ unchanged symbols from that code.

For linear codes, $\delta_\rho^{\elp}=d_{\lceil\rho\rceil}$, and the results reduce to GHW bounds in \cite{gruica2026linear}. The Hamming construction attains both general bounds for every initial code. The quasigroup construction shows that the same bounds are tight for nonlinear codes with nonintegral dimensions, including at a nonintegral entropy threshold. 

\appendix
\subsection{Proof of Lemma \ref{lem8}}\label{app1}

\begin{IEEEproof}
Let $T=\bigsqcup_{j\in I}\mathcal U_j$. In this proof, all random variables are defined on $\mathcal C_F$ equipped with the uniform distribution. Fix any
\begin{equation*}
    \mathbf u\in\mathcal C_F|_T.
\end{equation*}
Since $\sigma$ is bijective, there is an input tuple $(\mathbf v_1,\ldots,\mathbf v_t)$ whose final codeword has projection $\mathbf u$ on $T$. Fix the initial codewords $\mathbf v_j$ for all $j\in I$ and vary the remaining initial codewords. This produces
\begin{equation*}
    \prod_{j\in[t]\setminus I}M_{I_j} =q^{k_{-I}}
\end{equation*}
distinct final codewords. Since $\sigma$ is bijective, their projections on $\Omega_F\setminus T$ are distinct. Therefore,
\begin{equation*}
    \left| \left\{ \mathbf c|_{\Omega_F\setminus T}: \mathbf c\in\mathcal C_F,\  \mathbf c|_T=\mathbf u \right\} \right| \ge q^{k_{-I}}.
\end{equation*}
Define $\mathcal C_F(T,\mathbf u)=\{\mathbf c\in\mathcal C_F:\mathbf c|_T=\mathbf u\}$. For every $\mathbf c\in\mathcal C_F(T,\mathbf u)$, uniformity on $\mathcal C_F$ gives
\begin{align*}
    &\Pr\!\left(X^{(\mathcal C_F)}=\mathbf c
    \middle|X_T^{(\mathcal C_F)}=\mathbf u\right)\\
    &=\frac{1/|\mathcal C_F|}{|\mathcal C_F(T,\mathbf u)|/|\mathcal C_F|}\\
    &=\frac{1}{|\mathcal C_F(T,\mathbf u)|}.
\end{align*}
Thus the conditional distribution is uniform on $\mathcal C_F(T,\mathbf u)$. Since its codewords agree on $T$, their projections on $\Omega_F\setminus T$ are distinct. The corresponding conditional projection is therefore uniform on a set of at least $q^{k_{-I}}$ elements. Consequently, for every $\mathbf u\in\mathcal C_F|_T$,
\begin{equation*}
    H\left( X_{\Omega_F\setminus T}^{(\mathcal C_F)} \,\middle|\, X_T^{(\mathcal C_F)}=\mathbf u \right) \ge k_{-I}.
\end{equation*}
Averaging over all possible values of $X_T^{(\mathcal C_F)}$ gives
\begin{align*}
    &h_{n_F-|T|}(\mathcal C_F) \\
    &\ge H\left( X_{\Omega_F\setminus T}^{(\mathcal C_F)} \mid X_T^{(\mathcal C_F)} \right)\\
    &= \sum_{\mathbf u\in\mathcal C_F|_T} \Pr\left( X_T^{(\mathcal C_F)}=\mathbf u \right) H\left( X_{\Omega_F\setminus T}^{(\mathcal C_F)} \,\middle|\, X_T^{(\mathcal C_F)}=\mathbf u \right)\\
    &\ge k_{-I}.
\end{align*}
It follows from Definition \ref{def4} that
\begin{equation*}
    \delta_{k_{-I}}^{\elp}(\mathcal C_F) \le n_F-|T|.
\end{equation*}
Since the unchanged sets are pairwise disjoint,
\begin{equation*}
    |T|=\sum_{j\in I}|\mathcal U_j|.
\end{equation*}
Therefore,
\begin{equation*}
    \sum_{j\in I}|\mathcal U_j|\le n_F-\delta_{k_{-I}}^{\elp}(\mathcal C_F).
\end{equation*}
Combining this inequality with
\begin{equation*}
    \sum_{j\in I}|\mathcal U_j|\le \sum_{j\in I}n_{I_j}
\end{equation*}
proves the result.
\end{IEEEproof}

\subsection{An Example with $\alpha=\beta$}\label{app2}

Let $\mathcal C_{I_i}=\{00,11\}$ for $i\in[3]$, and define the final code by the bijection
\begin{equation*}
	\sigma(aa,bb,cc)=\bigl(a,b,c,\neg(a\lor b\lor c)\bigr),\ a,b,c\in\{0,1\}.
\end{equation*}
The last coordinate is a nonlinear Boolean function. More explicitly, both $(0,0,0,1)$ and $(0,0,1,0)$ belong to $\mathcal C_F$, whereas their binary sum $(0,0,1,1)$ does not. Hence, $\mathcal C_F$ is not closed under addition and is therefore nonlinear.

For all three Appendices \ref{app2}-\ref{app4}, write
\begin{equation*}
 	\mu_j(\mathcal C_F)=\min_{\substack{J\subseteq[n_F]\\ |J|=j}}H\bigl(X_J^{(\mathcal C_F)}\bigr).
\end{equation*}
Since $H(X_{[n_F]}^{(\mathcal C_F)})=3$, Definition~\ref{def2} gives
\begin{equation*}
 	h_i(\mathcal C_F)=3-\mu_{n_F-i}(\mathcal C_F).
\end{equation*}
To keep the calculations compact, Table~\ref{tab2} combines the values of $\mu_j$, $h_i$, and the corresponding ELP inverses for the three examples in Appendices \ref{app2}-\ref{app4}. The entries follow by grouping identical projections of the eight equiprobable codewords. We use
\begin{align*}
 	\gamma&=-\frac18\log_2\frac18-\frac78\log_2\frac78\approx0.5436,\\
	\lambda&=-\frac12\log_2\frac12-\frac38\log_2\frac38-\frac18\log_2\frac18 \approx1.4056.
\end{align*}

\begin{table*}[t]
	\centering
	\caption{Combined ELP calculations for the three appendix examples}
	\label{tab2}
	\scriptsize
	\renewcommand{\arraystretch}{1.5}
	\begin{tabular*}{\textwidth}{@{\extracolsep{\fill}}ccccc@{}}
		\toprule
		Case & $(\mu_0,\mu_1,\ldots,\mu_{n_F})$ &
		$(h_0,h_1,\ldots,h_{n_F})$ &
		$(\delta^{\elp}_1,\delta^{\elp}_2)$ &
		$(\alpha,\beta)$\\
		\midrule
		$\alpha=\beta$ &
		$(0,\gamma,\lambda,\frac94,3)$ &
		$(0,\frac34,3-\lambda,3-\gamma,3)$ &
		$(2,3)$ & $(2,2)$\\
		\midrule
		$\alpha>\beta$ &
		$(0,\gamma,\gamma,\lambda,\frac94,3)$ &
		$(0,\frac34,3-\lambda,3-\gamma,3-\gamma,3)$ &
		$(2,3)$ & $(4,3)$\\
		\midrule
		$\alpha<\beta$ &
		$(0,1,2,2,3,3,3,3)$ &
		$(0,0,0,0,1,1,2,3)$ &
		$(4,6)$ & $(2,3)$\\
		\bottomrule
	\end{tabular*}
\end{table*}

Here $n_F=4$. Because every initial code has size two while $M_F=8$, we have $k_{-1}=k_{-2}=2$ and $k_{-\{1,2\}}=1$. Consequently,
\begin{align*}
 \alpha&=2n_F-2\delta^{\elp}_2=8-6=2,\\
 \beta&=n_F-\delta^{\elp}_1=4-2=2.
\end{align*}
Hence the separate and joint ELP expressions coincide. Notice that $n_{I_1}+n_{I_2}=4>\alpha=\beta$, so this equality is not caused by the outer minimum with the total initial length.

\subsection{An Example with $\alpha>\beta$}\label{app3}

Again let $\mathcal C_{I_i}=\{00,11\}$ for $i\in[3]$, and define
\begin{equation*}
 \sigma(aa,bb,cc)=\bigl(a,b,c,a\land b\land c,\neg(a\land b\land c)\bigr).
\end{equation*}
This is a bijection onto a nonlinear code. The second row of Table~\ref{tab2} gives $\delta^{\elp}_1=2$ and $\delta^{\elp}_2=3$. Therefore,
\begin{align*}
 \alpha&=2n_F-2\delta^{\elp}_2=10-6=4,\\
 \beta&=n_F-\delta^{\elp}_1=5-2=3,
\end{align*}
and hence $\alpha>\beta$. In this example, the joint bound of Lemma~\ref{lem8} is strictly stronger than the sum of the two bounds from Lemma~\ref{lem7}.

\subsection{An Example with $\alpha<\beta$}\label{app4}

Let $\mathcal C_{I_i}=\{00,11\}$ for $i\in[3]$, and use the bijection
\begin{equation*}
 \sigma(aa,bb,cc)= (a,b,c,a\oplus b,a\oplus c,b\oplus c,a\oplus b\oplus c).
\end{equation*}
The final code is the binary $[7,3,4]$ Simplex code. For this linear code, the values of $\mu_j$ in the third row of Table~\ref{tab2} can alternatively be obtained from the minimum ranks of the corresponding sets of columns in a generator matrix. Thus $\delta^{\elp}_1=4$ and $\delta^{\elp}_2=6$, and
\begin{align*}
 \alpha&=2n_F-2\delta^{\elp}_2=14-12=2,\\
 \beta&=n_F-\delta^{\elp}_1=7-4=3.
\end{align*}
Hence $\alpha<\beta$. Here the sum of the two separate ELP expressions is
stronger than the joint expression.

\subsection{A Nonlinear Conversion with Noninteger Dimensions}\label{app5}

Let the alphabet be $\mathbb F_4=\{0,1,\omega,\omega^2\}$ and take
\begin{equation*}
 	\mathcal C_{I_1}=\mathcal C_{I_2}=\{00,11\}.
\end{equation*}
These are nonlinear codes over $\mathbb F_4$. Define
\begin{equation*}
 	\sigma(aa,bb)=(a,b),\qquad a,b\in\{0,1\}.
\end{equation*}
Then
\begin{equation*}
 \mathcal C_F=\{00,01,10,11\}\subseteq\mathbb F_4^2.
\end{equation*}
The final code is also nonlinear over $\mathbb F_4$, since $(10)\in\mathcal C_F$ but $\omega(10)=(\omega,0)\notin\mathcal C_F$. The map $\sigma$ is bijective, reduces the total length from four to two, and leaves one symbol from each initial word unchanged. Thus
\begin{equation*}
 	n_F=2,\quad d(\mathcal C_F)=1,\quad
 	|\mathcal U_1|=|\mathcal U_2|=1,\quad |\mathcal U|=2.
\end{equation*}
Moreover,
\begin{equation*}
	\log_4M_{I_1}=\log_4M_{I_2}=\frac{1}{2},
 	\qquad k_{-1}=k_{-2}=\frac{1}{2}.
\end{equation*}
Corollary \ref{cor1} gives, for $i=1,2$,
\begin{equation*}
 	|\mathcal U_i| \le2-1-\left\lceil\frac12\right\rceil+1=1,
\end{equation*}
and Lemma \ref{lem9} gives
\begin{equation*}
 	|\mathcal U| \le\left\lfloor4-\left\lceil\frac12\right\rceil -\left\lceil\frac12\right\rceil\right\rfloor=2.
\end{equation*}
Both bounds are attained. By contrast, omitting the separate ceiling terms would give only
\begin{equation*}
 	|\mathcal U|\le4-\log_4M_F=3.
\end{equation*}
The linear results in \cite{gruica2026linear} do not apply because neither initial code nor the final code is linear over $\mathbb F_4$, and the combinatorial dimensions of the two initial codes are nonintegral.

\subsection{Proof of the Noncontainment Statement in Proposition \ref{prop2}}\label{app6}

\begin{IEEEproof}
The table in \eqref{eqn50} gives
\begin{equation}\label{eqn59}
	0*0=0,\ 0*1=1,\ 1*0=1,\ 1*1=0.
\end{equation}
Fix an initial code $\mathcal C_{I_i}$ and set $\mathbf{x}_i=(x_{i,1},x_{i,2},0,\dots,0)$. Since zero is an identity element for $*$, the four choices $(x_{i,1},x_{i,2})\in\{0,1\}^2$ give four codewords that are constant outside coordinates $1$, $2$, and $m+1$ and whose projections onto these coordinates are
\begin{equation}\label{eqn60}
	(x_{i,1},x_{i,2},p_i)= (0,0,0),\ (0,1,1),\ (1,0,1),\ (1,1,0).
\end{equation}
Suppose that an equivalent image of $\mathcal C_{I_i}$ were contained in a linear $[m+1,m,2]_5$ code. Such a linear code has a parity-check equation
\begin{equation}\label{eqn61}
	\sum_{j=1}^{m+1}\beta_j y_j=0,
	\ \beta_j\in\mathbb F_5^*,
\end{equation}
with nonzero coefficients $\beta_j\ne 0$ for $j\in[m+1]$. Under the equivalence in \eqref{eqn52}, the three varying coordinates in \eqref{eqn60} are sent to three distinct positions, say $r_1,r_2,r_3$, with alphabet permutations $\tau_{r_1},\tau_{r_2},\tau_{r_3}$, i.e., 
\begin{equation*}
	(x_{i,1},x_{i,2},p_i) \to (\tau_{r_1}(x_{i,1}),\tau_{r_2}(x_{i,2}),\tau_{r_3}(p_i)).
\end{equation*}
Put
\begin{equation*}
	\Delta_j=\tau_{r_j}(1)-\tau_{r_j}(0)\ne0,
	\ j\in[3].
\end{equation*}
Consider
\begin{align*}
	(0,1,1),(0,0,0): \beta_{r_2}\Delta_2+\beta_{r_3}\Delta_3&=0,\\
	(1,0,1),(0,0,0): \beta_{r_1}\Delta_1+\beta_{r_3}\Delta_3&=0,\\
	(1,1,0),(0,0,0): \beta_{r_1}\Delta_1+\beta_{r_2}\Delta_2&=0.
\end{align*}
Thus
\begin{equation*}
	2\beta_{r_3}\Delta_3=0,
\end{equation*}
which is impossible in $\mathbb F_5$ since $\beta_{r_3}\ne 0$ and $\Delta_3\ne0$. Thus no code equivalent to a linear $[m+1,m,2]_5$ code contains $\mathcal C_{I_i}$.

For the final code, fix a block $i$, set every systematic coordinate other than $x_{i,1}$ and $x_{i,2}$ to zero, and again let $(x_{i,1},x_{i,2})$ range over $\{0,1\}^2$. Equations \eqref{eqn51}, \eqref{eqn54}, and \eqref{eqn59} give the same four projections in \eqref{eqn60} on coordinates $(i,1)$, $(i,2)$, and $mt+1$. Repeating the preceding parity-check argument with a linear $[mt+1,mt,2]_5$ code gives the same contradiction.
\end{IEEEproof}

\bibliographystyle{IEEEtran}
\bibliography{GeBib}

@misc{gruica2026linear,
  author    = {Gruica, Anina and Jany, Benjamin and Kruglik, Stanislav},
  title     = {Linear Code Conversion in the Merge Regime: General Bounds and Reed-Muller Constructions},
  year      = {2026},
  note      = {{arXiv}:2606.27179},
  doi       = {10.48550/arXiv.2606.27179},
  url       = {http://arxiv.org/abs/2606.27179}
}

@article{singhvi2026lower,
  author    = {Singhvi, Shubhransh and Chopra, Saransh and Rashmi, K. V.},
  title     = {Lower Bounds on Code Conversion Bandwidth via a Novel Information-theoretic Approach},
  journal   = {IEEE Transactions on Information Theory},
  year      = {2026},
  doi       = {10.1109/TIT.2026.3726356},
  note      = {early access}
}

@article{ge2026locally,
  author    = {Ge, Songping and Cai, Han and Tang, Xiaohu},
  title     = {Locally Repairable Convertible Codes: Improved Lower Bound and General Construction},
  journal   = {{IEEE} Transactions on Information Theory},
  year      = {2026},
  volume    = {72},
  number    = {5},
  pages     = {2915--2930},
  doi       = {10.1109/TIT.2026.3675568}
}

@article{ge2026mds,
  author    = {Ge, Songping and Cai, Han and Tang, Xiaohu},
  title     = {{MDS} Irregular Convertible Code},
  journal   = {IEEE Transactions on Information Theory},
  year      = {2026},
  volume    = {72},
  number    = {9},
  pages     = {6422--6434},
  doi       = {10.1109/TIT.2026.3714887}
}

@article{kong2024locally,
  author    = {Kong, Xiangliang},
  title     = {Locally Repairable Convertible Codes with Optimal Access Costs},
  journal   = {IEEE Transactions on Information Theory},
  year      = {2024},
  volume    = {70},
  number    = {9},
  pages     = {6239--6257},
  doi       = {10.1109/TIT.2024.3435346}
}

@article{krotov2008decomposability,
  author    = {Denis S. Krotov},
  title     = {On Decomposability of 4-Ary Distance 2 {MDS} Codes, Double-Codes, and $n$-Quasigroups of Order 4},
  journal   = {Discrete Mathematics},
  year      = {2008},
  volume    = {308},
  number    = {15},
  pages     = {3322--3334},
  doi       = {10.1016/j.disc.2007.06.038}
}

@article{krotov2016transitive,
  author    = {Denis S. Krotov and Vladimir N. Potapov},
  title     = {Constructions of Transitive Latin Hypercubes},
  journal   = {European Journal of Combinatorics},
  year      = {2016},
  volume    = {54},
  pages     = {59--74},
  doi       = {10.1016/j.ejc.2015.12.001}
}

@article{maturana2022convertible,
  author    = {Maturana, Francisco and Rashmi, K. V.},
  title     = {Convertible Codes: Enabling Efficient Conversion of Coded Data in Distributed Storage},
  journal   = {IEEE Transactions on Information Theory},
  year      = {2022},
  volume    = {68},
  number    = {7},
  pages     = {4392--4407},
  doi       = {10.1109/TIT.2022.3155972}
}

@article{maturana2023bandwidth,
  author    = {Maturana, Francisco and Rashmi, K. V.},
  title     = {Bandwidth Cost of Code Conversions in Distributed Storage: Fundamental Limits and Optimal Constructions},
  journal   = {IEEE Transactions on Information Theory},
  year      = {2023},
  volume    = {69},
  number    = {8},
  pages     = {4993--5008},
  doi       = {10.1109/TIT.2023.3265512}
}

@article{mckay2007small,
  author    = {Brendan D. McKay and Alison Meynert and Wendy Myrvold},
  title     = {Small Latin Squares, Quasigroups, and Loops},
  journal   = {Journal of Combinatorial Designs},
  year      = {2007},
  volume    = {15},
  number    = {2},
  pages     = {98--119},
  doi       = {10.1002/jcd.20105}
}

@article{ramkumar2026mds,
  author    = {Ramkumar, Vinayak and Kong, Xiangliang and Sai, G. Yeswanth and Vajha, Myna and Krishnan, M. Nikhil},
  title     = {On {MDS} Convertible Codes in the Merge Regime},
  journal   = {{IEEE} Transactions on Information Theory},
  year      = {2026},
  volume    = {72},
  number    = {5},
  pages     = {2963--2977},
  doi       = {10.1109/TIT.2026.3675319}
}

@article{reuven1998entropy,
  author    = {Reuven, Ilan and Be'ery, Yair},
  title     = {Entropy/Length Profiles, Bounds on the Minimal Covering of Bipartite Graphs, and Trellis Complexity of Nonlinear Codes},
  journal   = {IEEE Transactions on Information Theory},
  year      = {1998},
  volume    = {44},
  number    = {2},
  pages     = {580--598},
  doi       = {10.1109/18.661506}
}

@article{reuven1999generalized,
  author    = {Reuven, Ilan and Be'ery, Yair},
  title     = {Generalized Hamming Weights of Nonlinear Codes and the Relation to the $\mathbb{Z}_4$-Linear Representation},
  journal   = {IEEE Transactions on Information Theory},
  year      = {1999},
  volume    = {45},
  number    = {2},
  pages     = {713--720},
  doi       = {10.1109/18.749020}
}

@article{shi2026bounds,
  author    = {Shi, Haoming and Fang, Weijun and Gao, Yuan},
  title     = {Bounds and Optimal Constructions of Generalized Merge-Convertible Codes for Code Conversion Into {LRCs}},
  journal   = {{IEEE} Transactions on Information Theory},
  year      = {2026},
  volume    = {72},
  number    = {5},
  pages     = {2841--2860},
  doi       = {10.1109/TIT.2026.3668262}
}

@article{wang2026lower,
  author    = {Wang, Lewen and Hu, Sihuang},
  title     = {Lower Bounds on Conversion Bandwidth for {{MDS}} Convertible Codes in Split Regime},
  journal   = {{IEEE} Transactions on Information Theory},
  year      = {2026},
  volume    = {72},
  number    = {7},
  pages     = {4850--4859},
  doi       = {10.1109/TIT.2026.3692929}
}

@article{wei1991generalized,
  author    = {Wei, Victor K.},
  title     = {Generalized Hamming Weights for Linear Codes},
  journal   = {IEEE Transactions on Information Theory},
  year      = {1991},
  volume    = {37},
  number    = {5},
  pages     = {1412--1418},
  doi       = {10.1109/18.133259},
  month     = {sep}
}

@inproceedings{chopra2024low,
  author    = {Chopra, Saransh and Maturana, Francisco and Rashmi, K. V.},
  title     = {On Low Field Size Constructions of Access-Optimal Convertible Codes},
  booktitle = {2024 {{IEEE International Symposium}} on {{Information Theory}} ({{ISIT}})},
  year      = {2024},
  pages     = {1456--1461},
  doi       = {10.1109/ISIT57864.2024.10619440}
}

@inproceedings{ghemawat2003google,
  author    = {Ghemawat, Sanjay and Gobioff, Howard and Leung, Shun-Tak},
  title     = {The {{Google}} File System},
  booktitle = {Proceedings of the Nineteenth {{ACM}} Symposium on {{Operating}} Systems Principles},
  year      = {2003},
  pages     = {29--43},
  doi       = {10.1145/1165389.945450}
}

@inproceedings{huang2012erasure,
  author    = {Huang, Cheng and Simitci, Huseyin and Xu, Yikang and Ogus, Aaron and Calder, Brad and Gopalan, Parikshit and Li, Jin and Yekhanin, Sergey},
  title     = {Erasure coding in {Windows} {Azure} storage},
  booktitle = {2012 {{USENIX Annual Technical Conference}} ({{USENIX ATC}} 12)},
  year      = {2012},
  pages     = {15--26}
}

@inproceedings{kadekodi2019cluster,
  author    = {Kadekodi, Saurabh and Rashmi, K. V. and Ganger, Gregory R},
  title     = {Cluster Storage Systems Gotta Have {HeART}: Improving Storage Efficiency by Exploiting Disk-Reliability Heterogeneity},
  booktitle = {17th USENIX Conference on File and Storage Technologies (FAST 19)},
  year      = {2019},
  pages     = {345--358}
}

@inproceedings{maturana2020accessoptimal,
  author    = {Maturana, Francisco and Mukka, V. S. Chaitanya and Rashmi, K. V.},
  title     = {Access-Optimal Linear {MDS} Convertible Codes for All Parameters},
  booktitle = {2020 IEEE International Symposium on Information Theory (ISIT)},
  year      = {2020},
  pages     = {577--582},
  doi       = {10.1109/ISIT44484.2020.9173947}
}

@inproceedings{maturana2020convertible,
  author    = {Maturana, Francisco and Rashmi, K. V.},
  title     = {Convertible Codes: New Class of Codes for Efficient Conversion of Coded Data in Distributed Storage},
  booktitle = {11th Innovations in Theoretical Computer Science Conference (ITCS 2020)},
  year      = {2020},
  pages     = {66:1--66:26},
  doi       = {10.4230/LIPIcs.ITCS.2020.66}
}

@inproceedings{maturana2022bandwidth,
  author    = {Maturana, Francisco and Rashmi, K. V.},
  title     = {Bandwidth Cost of Code Conversions in the Split Regime},
  booktitle = {2022 IEEE International Symposium on Information Theory (ISIT)},
  year      = {2022},
  pages     = {3262--3267},
  doi       = {10.1109/ISIT50566.2022.9834604}
}

@inproceedings{maturana2023locally,
  author    = {Maturana, Francisco and Rashmi, K. V.},
  title     = {Locally {{Repairable Convertible Codes}}: {{Erasure Codes}} for {{Efficient Repair}} and {{Conversion}}},
  booktitle = {2023 IEEE International Symposium on Information Theory (ISIT)},
  year      = {2023},
  pages     = {2033--2038},
  doi       = {10.1109/ISIT54713.2023.10206604}
}

@book{huffman2010fundamentals,
  author    = {Huffman, W Cary and Pless, Vera},
  title     = {Fundamentals of Error-correcting Codes},
  publisher = {Cambridge University Press},
  year      = {2010},
  doi       = {10.1017/CBO9780511807077}
}

\end{document}